\documentclass{article}

\usepackage{a4wide}
\usepackage[utf8x]{inputenc}
\usepackage[english]{babel}
\usepackage[T1]{fontenc}

\usepackage{amsmath,amsfonts,amssymb,amsthm}

\usepackage{mathrsfs}

\usepackage{enumitem}

\usepackage[svgnames]{xcolor}

\usepackage{bbold}

\usepackage{tikz}
\usetikzlibrary{fit}
\usetikzlibrary{arrows.meta}

\usepackage{tikz-cd}

\usepackage{witharrows}

\usepackage{hyperref}

\theoremstyle{definition}
\newtheorem{definition}{Definition}[section]
\newtheorem{proposition}[definition]{Proposition}
\newtheorem{example}[definition]{Example}
\newtheorem{lemma}[definition]{Lemma}
\newtheorem{theorem}[definition]{Theorem}
\newtheorem{corollary}[definition]{Corollary}

\theoremstyle{remark}
\newtheorem{remark}{Remark}

\definecolor{ColBlack}{RGB}{0,0,0} 
\definecolor{ColWhite}{RGB}{255,255,255} 
\definecolor{ColBlue}{RGB}{13,20,120} 
\definecolor{ColCyan}{RGB}{51,122,121} 
\definecolor{ColLCyan}{RGB}{217,253,253} 

\tikzstyle{Centering}=[{baseline={([yshift=-0.5ex]current bounding box.center)}}]
\tikzstyle{EdgeGraph}=[ColBlue,cap=round,very thick]
\tikzstyle{NodeHyper}=[circle, draw=black, fill= black!20,inner sep=1pt,minimum size=3mm, thick, font=\scriptsize]
\tikzstyle{EdgeHyper}=[black, cap=roud, thick]
\tikzstyle{ArcHyper}=[EdgeHyper,->]
\tikzstyle{NodeFree}=[circle,draw=black,fill=white,inner sep=1pt,minimum size=6mm, thick,font=\scriptsize]
\tikzstyle{EdgeFree}=[black,cap=round,very thick]

\tikzset{
  pics/carc/.style args={#1:#2:#3}{
    code={
      \draw[pic actions] (#1:#3) arc(#1:#2:#3);
    }
  }
}

\newcommand{\noemp}{\not = \emptyset}                                   

\newcommand{\set}[1]{\left\{#1\right\}}                                 
\newcommand{\N}{\mathbb{N}}

\newcommand{\R}{\mathbb{R}}
\newcommand{\K}{\mathbb{K}}

\newcommand{\dom}{\text{D}}                                             
\newcommand{\p}{\mathcal{P}}                                            
\newcommand{\acyclic}{\mathcal{A}}                                      
\newcommand{\conv}{\text{conv}}                                         
\newcommand{\scone}[1]{\mathcal{N}^o_{#1}}                              
\newcommand{\cone}[1]{\mathcal{N}_{#1}}                                 
\newcommand{\colorings}[1]{\overline{C}_{#1}}                           
\newcommand{\scolorings}[1]{{C}_{#1}}                                   
\newcommand{\I}{I}                                                      
\newcommand{\NI}{NI}                                                    
\newcommand{\bfor}[1]{{\operatorname{\mathit{#1-\mathcal{F}}}}}         

\newcommand{\id}{Id}                                                    
\newcommand{\card}[1]{|#1|}                                             
\newcommand{\twist}{\tau}                                               
\newcommand{\antipode}{\mathsf{S}}                                      
\newcommand{\1}{\mathbb{1}}                                             
\newcommand{\bij}{\text{bij}}                                           
\newcommand{\forget}{\text{forget}_{\emptyset}}                         
\newcommand{\cc}{\text{cc}}                                             

\newcommand{\spe}[1]{#1}                                                
\newcommand{\SG}{\spe{SG}}                                              
\newcommand{\SHG}{\spe{SHG}}                                            
\newcommand{\G}{\spe{G}}                                                
\newcommand{\HG}{\spe{HG}}                                              
\newcommand{\GP}{\spe{GP}}                                              
\newcommand{\HGP}{\spe{HGP}}                                            
\newcommand{\SC}{\spe{SC}}                                              
\newcommand{\BS}{\spe{BS}}                                              
\newcommand{\Path}{\spe{Path}}                                          
\newcommand{\Dcomp}{\spe{Dcomp}}                                        
\newcommand{\Part}{\spe{\Pi}}                                           

\title{Combinatorial expressions of Hopf polynomial invariants}

\author{Théo Karaboghossian\\ 
    {\small Univ. Bordeaux, Bordeaux INP, CNRS, LaBRI, UMR5800, F-33400 Talence, France}\\
    {\small theo.karaboghossian@univ-littoral.fr}\\
    {\small ORCID:  0000-0001-8483-494X}
}
\date{}

\begin{document}

\maketitle

\begin{abstract}
In 2017 Aguiar and Ardila provided a generic way to construct polynomial invariants of combinatorial objects using the notions of
Hopf monoids and characters of Hopf monoids. 
The polynomials constructed this way are often subject to reciprocity theorems depending on the antipode of the associated Hopf monoid, i.e. 
while they are defined over positive integers, 
it is possible to find them a combinatorial interpretation over negative integers.
In the same article Aguiar and Ardila then give a cancellation-free grouping-free formula for the antipode on generalized permutahedra and apply 
their constructions over some examples.
In this work, we give a combinatorial interpretation of these polynomials over both positive integers and negative integers for the Hopf monoids 
of generalized permutahedra and hypergraphs and for every character on these two Hopf monoids. 
In the case of hypergraphs, we present two different proofs for the interpretation on negative integers,
one using Aguiar and Ardila's antipode formula 
and one similar to the way Aval et al. defined a chromatic polynomial for hypergraphs in arXiv:1806.08546. 
We then deduce similar results on other combinatorial objects including graphs, simplicial complexes and building sets.
\end{abstract}
\bigskip
{\small
    \\
    \noindent\textbf{MSC classes:} 16T30, 05E99, 05C15, 05C31, 05C65, 05E45.
}



\tableofcontents

\section{Introduction}\label{intro}

The notion of Hopf monoid first appeared, albeit not under this name, in André Joyal's introductory work to the theory of species \cite{joy}, 
and was built in the continuity of works from Joni and Rota \cite{R} \cite{JR}. A full theory of Hopf monoids, as well as the denomination of Hopf monoid, 
was later developed by Aguiar and Mahajan in \cite{am2},\cite{am}.
As with Hopf algebras, Hopf monoids give an algebraic framework to deal with operations of merging (product) and splitting (co-product) combinatorial objects.  
Using the theory of Hopf monoids, Aguiar and Ardila showed in \cite{AA} how to construct polynomial invariants, recovering the chromatic 
polynomial of graphs, the Billera-Jia-Reiner polynomial of matroids and the strict order polynomial of posets. 
Their construction of such a polynomial invariant requires a character, that is a multiplicative linear map, so that on a fixed Hopf monoid,
we can construct a polynomial invariant $\chi^{\zeta}$ for every character $\zeta$ on this Hopf monoid.
Furthermore, Aguiar and Ardila give a way to compute these polynomial invariants on negative integers thereby recovering reciprocity 
theorems associated to these combinatorial objects. This computation is based on the notion of antipode of a Hopf monoid, which plays the role 
of inversion in Hopf theory. They complete this result by providing a cancellation-free and grouping-free formula for the antipode of all
the Hopf monoids they define: generalized permutahedra, matroids, posets, hypergraphs, simple hypergraphs, simplicial complexes, building sets, graphs 
(with two different Hopf structure), partitions and set of paths.
Their results were used in \cite{my1} to define a chromatic polynomial over hypergraphs subject to a reciprocity theorem and
in the author's Phd thesis \cite{thesis}, from which this work is in great part an extract.

In this work we apply Aguiar and Ardila's construction of polynomial invariants to the Hopf monoids of generalized permutahedra and all its 
sub-monoids defined in \cite{AA}: hypergraphs, simple hypergraphs, simplicial complexes, building sets, graphs, partitions and set of paths.
Our results are a direct extension to Aval et al. results \cite{my1} in the sense that we give a combinatorial description of $\chi^{\zeta}$, 
for all these Hopf monoids and all characters over them while they only do so for one character per Hopf monoid and do not consider generalized permutahedra. 
The descriptions we provide comes in the form of counting combinatorial objects in link with objects of the considered Hopf monoid.
For example, as we will show in Theorem~\ref{th_pol_inv_gp} and Theorem~\ref{th_chromatic_polynomial}, for $P$ a generalized 
permutahedron and $h$ a hypergraph 
we have the following general formulas for $\chi^{\zeta}$:
\begin{equation*}
    \chi^{\zeta}(P)(n) = \sum_{Q\leq P} \zeta(Q)\card{\scone{P}(Q)_n} \quad\text{ and }\quad  \chi^{\zeta}(h)(n) = \sum_{f\in\acyclic_h} \zeta(f(h))\card{\scolorings{h,f,n}},
\end{equation*}
where $\scone{P}$ is the closed normal fan of $P$, $\acyclic_h$ is the set of acyclic orientations of $h$ and $\scolorings{h,f,n}$ is the set of
colorings of $h$ with $n$ colors strictly compatible with the orientation $f$.
To obtain these kind of formulas for all the Hopf monoids we consider, we follow a similar process to the one Aguiar and Ardila used to obtain their antipode 
formulas: we begin by the case of the Hopf monoid of generalized permutahedra and deduce the results for other Hopf monoid from their injection in this first Hopf monoid.
Our results generalize many existing results in the literature e.g. \cite{my1}, \cite{sc} and \cite{gr}.

This work is organized as follows. In Section~\ref{intro} we give general context on Hopf monoids and some definitions on decompositions
which are standard objects when working with Hopf structures. 
In Section~\ref{poly_inv_gp} we present the Hopf monoid structure of generalized permutahedra defined in \cite{AA} and 
we provide a combinatorial interpretation to the polynomial invariants of the Hopf monoid of generalized permutahedra.
In Section~\ref{poly_inv_hyper}, we do the same with the Hopf monoid of hypergraphs but providing two proofs for the combinatorial
interpretation of the polynomial invariants over negative integers. Finally in Section~\ref{sec_other_hopf} we obtain similar results 
for the Hopf monoids defined in \cite{AA} over simple hypergraphs, simplicial complexes, building sets, graphs, partitions and set of paths.


\subsection*{Context}

\subsubsection*{Definitions and background on Hopf monoids} \label{sec_species}
In all this paper, otherwise stated, $V$ will always denote a finite set and $V_1$ and $V_2$ two disjoint sets such that
$V=V_1\sqcup V_2$. The letter $n$ always denotes a non negative integer
and we denote by $[n]$ the set $\set{1,\dots, n}$.
All vector spaces appearing in this paper are defined over a field of characteristic 0 denoted by $\K$.

We recall here basic definitions on Hopf monoids. We refer the reader to~\cite{BLL98} for a more 
information on the theory of species and to \cite{am} for more information on Hopf monoids.

\begin{definition}
    A {\em linear species} $\spe{S}$ consists of the following data:
    \begin{itemize}
       \item For each finite set $V$, a vector space $\spe{S}[V]$ of finite dimension,
       \item For each bijection of finite sets $\sigma: V\to V'$, a linear map $\spe{S}[\sigma]:\spe{S}[V]\to \spe{S}[V']$.
       These maps should be such that $\spe{S}[\sigma_1\circ\sigma_2] = \spe{S}[\sigma_1]\circ \spe{S}[\sigma_2]$ and $\spe{S}[\id] = \id$.
    \end{itemize}
\end{definition}

Let $\spe{R}$ and $\spe{S}$ be two linear species. The linear species $\spe{R}$ is a {\em linear sub-species} 
of $\spe{S}$ if $\spe{R}[V]$ is a sub-space of $\spe{S}[V]$ for every finite set $V$  and $\spe{R}[\sigma] = \spe{S}[\sigma]$ 
for every bijection of finite sets $\sigma$. A {\em morphism of linear species} from $\spe{R}$ to $\spe{S}$
is a collection of linear maps $f_V : \spe{R}[V] \to \spe{S}[V]$ such that for each bijection 
$\sigma: V\to V'$, we have $f_{V'}\circ \spe{R}[\sigma] = \spe{S}[\sigma]\circ f_V$. For easier reading,
we will often forget the index $V$.

We will use the term {\em species} to refer to linear species.

\begin{definition}
A {\em connected Hopf monoid in linear species} is a linear species $\spe{M}$ where 
$\spe{M}[\emptyset] = \K$ and which is equipped with a {\em product} and a {\em co-product}
linear maps:
\begin{equation*}
    \mu_{V_1,V_2}: \spe{M}[V_1]\otimes \spe{M}[V_2] \to \spe{M}[V_1\sqcup V_2], \qquad
    \Delta_{V_1,V_2}: \spe{M}[V_1\sqcup V_2] \to \spe{M}[V_1]\otimes \spe{M}[V_2],
\end{equation*}
which are subject to the following axioms, where we use the classical infix notation $\_\cdot\_$ for
the product.
\begin{itemize}
    \item {\em Naturality}. The maps $\mu : \spe{M}\cdot\spe{M}\to\spe{M}$ and 
    $\Delta:\spe{M}\to\spe{M}\cdot\spe{M}$ are morphisms of linear species.

    \item {\em Unitality}. For $x\in \spe{M}[V]$,
    $
        x\cdot 1_{\K} = x = 1_{\K}\cdot x,
    $
     \textit{i.e.} merging with the unit does not change our objects.

    \item {\em Co-unity}.
    $
        \Delta_{V,\emptyset} = \id\otimes 1_{\K}$ and  $\Delta_{\emptyset, V} = 1_{\K}\otimes\id,
    $
     \textit{i.e.} splitting on the empty set does not change our objects. 

    \item\label{hopf_monoid_associativity} {\em Associativity}. For $x\in\spe{M}[V_1]$, $y\in\spe{M}[V_2]$ and $z\in\spe{M}[V_3]$,
    $
        x\cdot(y\cdot z) = (x\cdot y)\cdot z,
    $
     \textit{i.e.} the results of merging three objects does not depend on the order in 
    which we merge them.

    \item\label{hop_monoid_co-associativity} {\em Co-associativity}.
    $
        \Delta_{V_1,V_2}\otimes\id\circ\Delta_{V_1\sqcup V_2, V_3} = \id\otimes\Delta_{V_2,V_3}\circ\Delta_{V_1,V_2\sqcup V_3},
    $
     \textit{i.e.} the result of splitting an object in three does not depend in the order in
    which we split it. 

    \item {\em Compatibility}. Let $V=V_1\sqcup V_2 = V_3\sqcup V_4$ and for $i\in\set{1,2}$ and 
    $j\in\set{3,4}$ denote by $V_{ij}$ the set $V_i\cap V_j$. Then we have the following commutative 
    diagram:
    \begin{equation*}
        \begin{tikzcd}
            \spe{M}[V_1]\otimes\spe{M}[V_2] \arrow[r, "\mu_{V_1,V_2}"] \arrow[d, "\Delta_{V_{13},V_{14}}\otimes\Delta_{V_{23},V_{24}}" swap] &
            \spe{M}[V] \arrow[r, "\Delta_{V_3,V_4}"] &
            \spe{M}[V_3]\otimes\spe{M}[V_4] \\
            \spe{M}[V_{13}]\otimes\spe{M}[V_{14}]\otimes\spe{M}[V_{23}]\otimes\spe{M}[V_{24}] \arrow[rr, "\id\otimes\twist\otimes"] & & 
            \spe{M}[V_{13}]\otimes\spe{M}[V_{23}]\otimes\spe{M}[V_{14}]\otimes\spe{M}[V_{24}] \arrow[u, "\mu_{V_{13},V_{23}}\otimes\mu_{V_{14},V_{24}}"]
        \end{tikzcd},
    \end{equation*}
    \textit{i.e.} merging and splitting is the same than splitting and merging.
\end{itemize}
\end{definition}

We use the term {\em Hopf monoid} to refer to a connected Hopf monoid in linear species. 

A Hopf {\em sub-monoid} of a Hopf monoid $\spe{M}$ is a sub-species of $\spe{M}$ stable under the product and co-product maps. 

A {\em morphism of Hopf monoids} is a morphism of linear species which preserves the products, 
co-products and unit $1_{\K}$.

The {\em co-opposite Hopf monoid} of a Hopf monoid $\spe{M}$ is the Hopf monoid $\spe{M}^{cop}$ over the same species than $\spe{M}$ 
and with the product and co-product defined by: $\mu_{V_1,V_2}^{\spe{M}^{cop}}=\mu_{V_1,V_2}^{\spe{M}}$ and 
$\Delta^{\spe{M}^{cop}}_{V_1,V_2}=\Delta^{\spe{M}}_{V_2,V_1}$.

\begin{remark}
    The axioms of associativity~\eqref{hopf_monoid_associativity} 
    and co-associativity~\eqref{hop_monoid_co-associativity} make it so that we can naturally extend the definitions of the 
    structure maps over any decomposition of $V$. For $V_1,\dots,V_n$ a decomposition of $V$:
    \begin{equation*}\label{iter_prod}
        \mu_{V_1,\dots,V_n}:\spe{M}[V_1]\otimes\dots\otimes\spe{M}[V_n]\to\spe{M}[V]\qquad\Delta_{V_1,\dots,V_n}:\spe{M}[V]\to\spe{M}[V_1]\otimes\dots\otimes\spe{M}[V_n],
    \end{equation*}
    are respectively defined by iterating any kind of maps of the form $\id^{\otimes k}\otimes\mu_{V_i,V_{i+1}}\otimes\id^{\otimes l}$ and of the form $\id^{\otimes k}\otimes\Delta_{V_i,V_{i+1}}\otimes\id^{\otimes l}$, as long as the domains and co-domains coincide.
\end{remark}

A {\em decomposition} of a finite set $V$ is a sequence of pairwise disjoint subsets $S = (V_1,\dots, V_l)$ such that $V = \bigsqcup_{i=1}^l V_i$. 
A {\em composition} of a finite set $V$ is a decomposition of $V$ without empty parts. We will write $S \vdash V$ for $S$ a decomposition of $V$, $S\vDash V$ 
if $S$ is a composition, $l(S) = l$ the {\em length} of a decomposition and $\card{S} = \card{V}$ the number of elements in the decomposition.

\begin{definition}\label{def_antipode}
    The antipode of a Hopf monoid $\spe{M}$ is the species morphism $\antipode:\spe{M}\to\spe{M}$
    defined by $\antipode_{\emptyset} = \id$ and for $V\noemp$,
    \begin{equation*}\label{takeuchi_formula}
        \antipode_V(x) = \sum_{V_1,\dots,V_n\vDash V} (-1)^n\mu_{V_1,\dots, V_n}\circ\Delta_{V_1,\dots,V_n}(x).
    \end{equation*}
    This formula is known as {\em Takeuchi's formula}. 
\end{definition}

A {\em character} $\zeta$ of a Hopf monoid $\spe{M}$ is a collection of linear 
maps $\set{\zeta_V:\spe{M}[V]\to\K}_V$ compatible with the product and sending the unit on the unit.

We say that a character is a {\em characteristic function} if it takes values in the set $\set{0,1}$.
A {\em discrete} element of a Hopf monoid $\spe{M}$ is an element which can be obtained as a product
of elements of size 1. \textit{i.e.} $x\in\spe{M}[V]$ is discrete if $V=\set{v_1,\dots, v_n}$ and there exists
$x_1\in\spe{M}[\set{v_1}],\dots,x_n\in\spe{M}[\set{v_n}]$ such that 
$x=\mu_{\set{v_1},\dots,\set{v_n}} x_1\otimes\dots\otimes x_n$.
The {\em basic character} of any Hopf monoid $\spe{M}$ is then the characteristic function of discrete
elements of $\spe{M}$. We will denote it by $\zeta_1$.

The main object of study of this paper are the following maps.

\begin{definition}\label{def_polynomial_invariant}
    Let $\spe{M}$ be a Hopf monoid and $\zeta$ a character on $\spe{M}$. For $x\in\spe{M}[V]$
    define the map
    \begin{equation*}\label{polynomial_invariant_formula}
        \chi^{(\spe{M},\zeta)}_V(x)(n) = \sum_{V=V_1\sqcup\dots\sqcup V_n} \zeta_V\circ\mu_{V_1,\dots, V_n}\circ\Delta_{V_1,\dots,V_n}(x).
    \end{equation*}
\end{definition}

Depending on how clear it is from the context, we will not specify $\zeta$ and/or $\spe{M}$
and use the notations $\chi^{\spe{M}}$, $\chi^{\zeta}$ or $\chi$ to designate the map thus 
defined. 
These maps have very interesting properties as shown in the following theorem and proposition.

\begin{theorem}[Proposition 16.1 and Proposition 16.2 in \cite{AA}]\label{th_polynomial_invariant}
    Let $\spe{M}$ be a Hopf monoid, $\zeta$ a character on $\spe{M}$ and $\chi^{(M,\zeta)}$ 
    be the collection of maps of Definition~\ref{def_polynomial_invariant}. Then 
    $\chi^{(M,\zeta)}_V(x)$ is a polynomial invariant in $n$ such that:
    \begin{enumerate}
        \item $\chi^{(M,\zeta)}_V(x)(1) = \zeta(x)$,
        \item $\chi^{(M,\zeta)}_{\emptyset} = 1$ and $\chi^{(M,\zeta)}_{V_1\sqcup V_2}(x\cdot y)) = \chi^{(M,\zeta)}_{V_1}(x)\chi^{(M,\zeta)}_{V_2}(y)$,
        \item\label{th_polynomial_invariant_3} $\chi^{(M,\zeta)}_V(x)(-n) = \chi^{(M,\zeta)}_V(\antipode_V(x))(n)$.
    \end{enumerate}
\end{theorem}

\begin{proposition}[Proposition 16.3 in \cite{AA}]\label{prop_polynomial_invariant}
    Let $\spe{M}$ and $\spe{N}$ be two Hopf monoids, $\zeta^{\spe{M}}$ and $\zeta^{\spe{N}}$
    be two character on these Hopf monoids and $\phi:\spe{M}\to\spe{N}$ be a morphism of Hopf
    monoids compatible with the characters: $\zeta^{\spe{N}}\circ \phi = \zeta^{\spe{M}}$.
    Then $\chi^{\spe{N}}\circ \phi = \chi^{\spe{M}}$.
\end{proposition}

\subsubsection*{Decomposition, compositions and colorings}\label{decomp}\label{decomp_op}
We briefly presented the species of decompositions and compositions in the previous subsection. We give here more details on these objects which will
often use thorough this paper.

A {\em coloring of $V$ with $[n]$} is a map from $V$ to $[n]$ and a {\em part ordering} of $V$ is a surjective map with domain $V$ 
and a co-domain of the form $[n]$.
There exists a canonical bijection between decompositions and colorings. In all this paper, 
we want to seamlessly pass from one notion to the other, so we give a few explanations on this bijection. Given an integer $n$, the canonical 
bijection between decompositions of $V$ of size $n$ and colorings of $V$ with $[n]$ is given by:
\begin{equation*}\begin{split}\label{bij_decomposition_colorings}
    b_{V,n}: \set{f:V\to [n]} &\to \set{P\in\Dcomp[V]\,|\, l(P)=n} \\
    f &\mapsto (f^{-1}(1),\dots,f^{-1}(n)).
\end{split}\end{equation*}
If it is clear from the context what are $V$ and $n$, we will write $b$ instead of $b_{V,n}$. If $P$ is a partition we will also refer to $b^{-1}(P)$ 
by $P$ so that instead of writing ``i such that $v\in P_i$'' and ``i and j such that $v\in P_i$, $v'\in P_j$ and $i<j$ '' we can just write $P(v)$ 
and $P(v)<P(v')$. Similarly, if $P$ is a function we will refer to $b(P)$ by $P$ so that $P_i = P^{-1}(i)$. Remark that $b_{V,n}$ induces a bijection 
between compositions and part orderings.
\medskip

Let us now give some usual operations and definitions over decompositions and compositions.
Let $V$ and $W$ be two disjoint sets and $P = (P_1,\dots,P_l)\vDash V$ and $Q = (Q_1,\dots, Q_k) \vDash W$ be two compositions. We call {\em product} of $P$ 
and $Q$ and we denote by $P\cdot Q$ the composition $(P_1,\dots,P_l,Q_1,\dots Q_k)$. We call {\em shuffle product} of $P$ and $Q$ the set $sh(P,Q)$ defined by
$sh(P,Q) = \set{R\,|\, (b^{-1}(R))_{|V}=b^{-1}(P)\text{ and }(b^{-1}(R))_{|V'}=b^{-1}(P')}$, where for $f$ a map with domain $V$ and $W\subseteq V$, the
map $f_{|W}$ is defined by $f_{|W}(v) = f(v)$ for all $v\in W$.

Let $P'=(P_{1,1},\dots,P_{1,k_1},P_{2,1},\dots,P_{2,k_2},\dots,P_{l,k_l})$ be another composition of $V$. We say that $P'$ {\em refines} 
$P$ and write $P'\prec P$ if $P_i = \bigcup_{j=1}^{k_i}P_{i,j}$ for $1\leq i\leq l$.



\section{Generalized permutahedra}\label{poly_inv_gp}
Using Aguiar and Ardila's results in \cite{AA}, we give here an explicit combinatorial interpretation of the polynomial
invariant in Definition~\ref{def_polynomial_invariant} and its reciprocity theorem over the Hopf monoid of generalized 
permutahedra. The two theorems of this section are direct generalizations of Proposition 17.3 and 17.4 \cite{AA}.

\subsection{Hopf monoid structure}
Let us first give the Hopf monoid structure of generalized permutahedra defined in \cite{AA} while introducing some classical
definitions over generalized permutahedra. We refer the reader interested in generalized permutahedra to the vast literature on 
the subject (see for example \cite{Post09}, \cite{Fuji05}, \cite{GMZ})

A {\em generalized permutahedron} in $\R V$ is a polytope in $\R V$ whose edges are parallel to the vectors $v-v'$ for $v,v'\in V$. We
denote by $\GP$ the species of generalized permutahedra.
We call the elements of the dual $(\R V)^*=\R^V$ {\em directions}. For $P\in \GP[V]$ and $y$ a direction in $\R^V$,
we call {\em maximum face of $P$ in the direction $y$} or {\em $y$-maximum face} of $P$ the generalized permutahedron 
$P_y = \set{p\in P\,|\, y(p)\geq y(q)\text{ for all }q\in P}$. In particular, the faces of dimension 1 are the {\em edges} of $P$. 

We denote by $Q\leq P$ for 
$Q$ a face of $P$ and $Q<P$ if additionally $Q\not=P$. We denote by $L(P)$ the {\em face lattice} of $P$, which is the poset of faces of $P$ 
ordered with the previously defined order. We denote by $[Q;P]$ the interval $\set{P'\,|\,Q\leq P'\leq P}$ and by $[Q;P[$ the same interval 
but with strict inequality on the right side. 

For each face $Q$ of $P$ define the {\em normal cones} as:
\begin{equation*}\begin{split}
    \scone{P}(Q) &= \set{y\in\R^V\,|\, P_y=Q}\\
    \cone{P}(Q) &= \set{y\in \R^V\,|\,\text{$Q$ is a face of $P_y$}}.
\end{split}\end{equation*}

\begin{proposition}[Theorem 3.15 in \cite{Fuji05}] 
    Let $V$ be a finite set, $W\subseteq V$ and $P$ be a generalized permutahedra in $\R V$. Denote by $\1_W$ the direction defined by
    $\1_W(v) = 1$ if $v\in W$ and $\1_W(v) = 0$ if $v\in V\setminus W$. Then there exists two generalized permutahedra $P|_W$ in $\R W$ and
    $P/_W$ in $\R V\setminus W$ such that $P_{\1_W} = P|_W + P/_W$.
\end{proposition}

\begin{theorem}[Theorem 5.3 in \cite{AA}]
    The species $\GP$ as a Hopf monoid structure given by
    \begin{align*}
        \mu_{V_1,V_2}:\GP[V_1]\otimes\GP[V_2]&\to\GP[V] & \Delta_{V_1,V_2}:\GP[V]&\to\GP[V_1]\otimes\GP[V_2] \\
        P\otimes Q &\mapsto P+Q & P &\mapsto P|_{V_1}\otimes P/_{V_1}, \nonumber
    \end{align*}
    where $P+Q=\set{p+q\,|\, p\in P, q\in Q}$ is the {\em Minkowski sum} of $P$ and $Q$
\end{theorem}

The Hopf monoid $\GP$ has the following formula for the antipode:

\begin{theorem}[Theorem 7.1 in \cite{AA}]\label{th_antipode_gp}
    The antipode of $\GP$ is given by the cancellation-free and grouping-free formula:
    \begin{equation*}
        \antipode_V(P)=\sum_{Q\leq P}(-1)^{\card{V}-\dim Q}Q.
    \end{equation*}
\end{theorem}

We end this subsection by mentioning a Hopf sub-monoid of $\GP$ which will use later on.
Denote by $\Delta_V = \conv(v\,|\,v\in V)$ the {\em standard simplex} of $\R V$. A {\em hypergraphic polytope} 
is a Minkowski sum of standard simplices. We denote by $\HGP$ the species of hypergraphic polytopes.

\begin{proposition}[Proposition 19.5 in \cite{AA}]
    The species $\HGP$ is a Hopf sub-monoid of $\GP$.
\end{proposition}


\subsection{Polynomial invariant and reciprocity theorem}

We begin by giving an interpretation of $\chi$ over non negative integers. Most definitions necessary to state our theorem were 
given in previous subsection, but we still need to introduce two simple notions. Let $P\in\GP[V]$ be a generalized 
permutahedron and remark that as maps from $V$ to $[n]$, colorings are elements of $\R^V$ \textit{i.e.} directions.
\begin{itemize}
    \item For $\zeta$ a character of $\GP$, $P$ is a {$\zeta$-face} if $\zeta(P)\not=0$. 
    \item For $Q$ a face of $P$ and $c$ a coloring of $V$, $Q$ and $c$ are said to be {\em strictly compatible} if $P_c = Q$. They are
    said to be {\em compatible} if $Q\leq P_c$. We respectively denote by $\scone{P}(Q)_n=[n]^V\cap\scone{P}(Q)$ and 
    $\cone{P}(Q)_n=[n]^V\cap\cone{P}(Q)$ the set of colorings with $[n]$ strictly compatible with $Q$ and compatible with $Q$.
\end{itemize}

\begin{theorem}\label{th_pol_inv_gp}
    Let $\zeta$ be a character of $\GP$, $V$ be a finite set and $P\in\GP[V]$ a generalized permutahedra. Then
    \begin{equation*}
        \chi^{\zeta}_V(P)(n) = \sum_{Q\leq P} \zeta(Q)\card{\scone{P}(Q)_n}.
    \end{equation*}
    In particular, if $\zeta$ is a characteristic function, then $\chi^{\zeta}_V(P)(n)$ is the number of strictly 
    compatible pairs of $\zeta$-faces of $P$ and colorings with $[n]$.
\end{theorem}

\begin{proof}
    We have:
    \begin{equation*}\begin{split}
        \chi^{\zeta}_V(P)(n) &= \sum_{V=V_1\sqcup\dots\sqcup V_n} \zeta_V\circ\mu_{V_1,\dots, V_n}\circ\Delta_{V_1,\dots,V_n}(P)\\
        &= \sum_{\text{$D\vdash V$, $l(D)=n$}} \zeta_V\circ\mu_D\circ\Delta_D(P)\\
        &= \sum_{\text{$D$ coloring of $V$ with $[n]$}} \zeta_V(P_D)\\
        &= \sum_{Q\leq P} \sum_{\substack{\text{$D$ coloring of $V$ with $[n]$}\\P_D=Q}}\zeta_V(Q)\\
        &= \sum_{Q\leq P} \zeta_V(Q) \card{\scone{P}(Q)_n}.
    \end{split}\end{equation*}
    where the third equality comes from \cite{AA} Proposition 5.4 which states that for $D=(V_1,\dots,V_n)$ a decomposition/coloring of $V$,
    the $D$-maximum face of $P$ is given by $P_D=\mu_D\circ\Delta_D(P)$.
\end{proof}

\begin{example}\label{ex_pol_inv_gp}
    Let $\zeta$ be the characteristic function of {\em graphic polytopes} \textit{i.e.} polytopes which can be written as a Minkowski sum of standard simplices of dimension 1.
    Let $P=\Delta_{e_1,e_2,e_3}+\Delta_{e_1,e_4}$:
    \begin{equation*}
        P = \enspace 
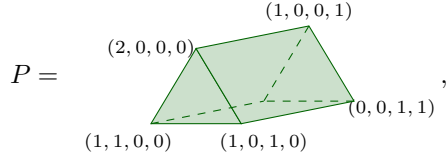
\begin{tikzpicture}[Centering]
            \node[NodeFree, draw=white](aa)at(-0.9,-0.25){$(1,1,0,0)$};
            \node[NodeFree, draw=white](aa)at(0.9,-0.25){$(1,0,1,0)$};
            \node[NodeFree, draw=white](aa)at(-0.6,1){$(2,0,0,0)$};
            \node[NodeFree, draw=white](aa)at(1.5,1.5){$(1,0,0,1)$};
            \node[NodeFree, draw=white](aa)at(2.6,0.2){$(0,0,1,1)$};
            \coordinate(a)at(-0.6,0);
            \coordinate(b)at(0.6,0);
            \coordinate(c)at(0,1);
            \coordinate(d)at(1.5,1.3);
            \coordinate(e)at(2.1,0.3);
            \filldraw[draw=DarkGreen, fill=DarkGreen!20] (a) -- (b) -- (c) -- cycle;
            \filldraw[draw=DarkGreen, fill=DarkGreen!20] (b) -- (c) -- (d) -- (e) -- cycle;
            \draw[DarkGreen, dashed](a)--(0.9,0.3);
            \draw[DarkGreen, dashed](d)--(0.9,0.3);
            \draw[DarkGreen, dashed](e)--(0.9,0.3);
        \end{tikzpicture},
    \end{equation*}
    where the hidden point is of coordinate $(0,1,0,1)$.
    From Theorem~\ref{th_pol_inv_gp} we know that $\chi^{\zeta}(P)(3)$ is the number of strictly compatible pairs of $\zeta$-faces of $P$ and colorings with $[2]$.
    The $\zeta$-faces of $P$ are its three rectangular faces which correspond to the sums $\Delta_{e_1,e_2} + \Delta_{e_1,e_4}$, $\Delta_{e_1,e_3} + \Delta_{e_1,e_4}$
    and $\Delta_{e_2,e_3} + \Delta_{e_1,e_4}$. Each of these faces is strictly compatible with exactly one coloring with $[2]$. These colorings are respectively 
    $2\mathbf{e_1^*}+2\mathbf{e_2^*}+\mathbf{e_3^*}+2\mathbf{e_4^*}$, $2\mathbf{e_1^*}+\mathbf{e_2^*}+2\mathbf{e_3^*}+2\mathbf{e_4^*}$ and 
    $\mathbf{e_1^*}+2\mathbf{e_2^*}+2\mathbf{e_3^*}+\mathbf{e_4^*}$. Hence we have $\chi^{\zeta}(P)(2)=3$.
\end{example}

The basic character $\zeta_1$ of $\GP$ is equal to 1 on points and 0 elsewhere. There is a particular interpretation for this
character which is presented in \cite{AA}. For $P\in\GP[V]$ a generalized permutahedra and $y$ a direction in $\R V$, $y$ is said to be {\em $P$-generic} if 
$P_y$ is a point.

\begin{corollary}[Theorem 9.2 (v) in \cite{BJR09} and Proposition 17.3 in \cite{AA}]
    Let $V$ be a finite set and $P\in\GP[V]$ a generalized permutahedra. Then $\chi^{\zeta_1}(P)(n)$ is the number of $P$-generic 
    colorings of $V$ with $[n]$.
\end{corollary}

\begin{proof}
    By definition the $\zeta_1$-faces are the points and a coloring $c$ is strictly compatible with a face $Q$ if $P_c=Q$.
    $\chi^{\zeta_1}(P)$ is then counting the number of colorings $c$ such that $P_c$ is a point \textit{i.e.} the number of $P$-generic colorings.
\end{proof}

We now give the reciprocity theorem associated with these polynomials. A character $\zeta$ of $\GP$ is said to be {\em even} if 
the $\zeta$-faces are of even dimension.
\begin{theorem}\label{rec_pol_inv_gp}
    Let $\zeta$ be a character of $\GP$, $V$ be a finite set and $P\in\GP[V]$ a generalized permutahedra then
    \begin{equation*}
        \chi^{\zeta}_V(h)(-n) = \sum_{Q\leq P} (-1)^{\card{V}-\dim Q}\zeta(Q)\card{\cone{P}(Q)_n}.
    \end{equation*}
    Furthermore if $\zeta$ is an even characteristic function then $(-1)^{\card{V}}\chi^{\zeta}(P)(-n)$ 
    is the number of compatible pairs of $\zeta$-faces of $P$ and colorings with $[n]$. In particular,
    $(-1)^{\card{V}}\chi^{\zeta}(P)(-1)$ is the number of $\zeta$-faces.
\end{theorem} 

\begin{proof}
    Lemma 17.2.1 in \cite{AA} state that
    \begin{equation*}
        \sum_{Q\leq P} (-1)^{\dim Q}Q_y = \sum_{Q\leq P_{-y}}(-1)^{\dim Q}Q.
    \end{equation*}
    Hence, for any direction $y$ we have:
    \begin{equation*}
        \sum_{Q\leq P} (-1)^{\dim Q}\zeta(Q_y) = \sum_{Q\leq P_{-y}}(-1)^{\dim Q}\zeta(Q).
    \end{equation*}

    Beginning with Theorem~\ref{th_antipode_gp} and Theorem~\ref{th_polynomial_invariant}.3 we have:
    \begin{equation*}\begin{split}
        \chi(P)(-n) &= \chi(\antipode(P))(n) = \chi(\sum_{Q\leq P}(-1)^{\card{V}-\dim Q}Q)(n) \\
        &= \sum_{Q\leq P}(-1)^{\card{V}-\dim Q}\chi(Q)(n) \\
        &= \sum_{Q\leq P}(-1)^{\card{V}-\dim Q}\sum_{R\leq Q} \zeta(R)\card{\scone{Q}(R)_n} \\
        &= \sum_{Q\leq P}(-1)^{\card{V}-\dim Q}\sum_{R\leq Q} \zeta(R)\sum_{\substack{y:V\to[n]\\ Q_y = R}} 1 \\
        &= \sum_{Q\leq P}(-1)^{\card{V}-\dim Q}\sum_{y:V\to[n]} \zeta(Q_y) \\
        &= \sum_{y:V\to[n]} (-1)^{\card{V}}\sum_{Q\leq P}(-1)^{\dim Q}\zeta(Q_y) \\
        &= \sum_{y:V\to[n]} (-1)^{\card{V}}\sum_{Q\leq P_{-y}}(-1)^{\dim Q}\zeta(Q) \\
        &= \sum_{Q\leq P}(-1)^{\card{V}-\dim Q}\zeta(Q)\sum_{\substack{y:V\to[n]\\ Q\leq P_{-y}}} 1 \\
        &= \sum_{Q\leq P}(-1)^{\card{V}-\dim Q}\zeta(Q) \card{\cone{P}(Q)_n},
    \end{split}\end{equation*}
    where the last equality comes from the fact that $P_{-y}=P_{n+1-y}$ and $n+1-y$ as co-domain $[n]$ is and only if
    $y$ has co-domain $[n]$.
\end{proof}

\begin{example}
    The character $\zeta$ of Example~\ref{ex_pol_inv_gp} is not even but the $\zeta$-faces of the polytope $P$ of Example~\ref{ex_pol_inv_gp} are of even dimension
    hence $(-1)^4\chi^{\zeta}(P)(-2)$ is equal to the number of compatible pairs of $\zeta$-faces of $P$ and colorings with $[2]$. For each rectangular
    face of $P$, its compatible colorings with $[2]$ are the one given in Example~\ref{ex_pol_inv_gp} along with the colorings 
    $\mathbf{e_1^*}+\mathbf{e_2^*}+\mathbf{e_3^*}+\mathbf{e_4^*}$ and $2\mathbf{e_1^*}+2\mathbf{e_2^*}+2\mathbf{e_3^*}+2\mathbf{e_4^*}$. Hence $\chi^{\zeta}(P)(-2)=9$.
\end{example}

\begin{corollary}[Theorem 6.3 and Theorem 9.2 (v) in \cite{BJR09} and Proposition 17.4 in \cite{AA}]
    Let $V$ be a finite set and $P\in\GP[V]$ a generalized permutahedra. Then, we have:
        \begin{equation*}\label{cor_poly_gp}
            (-1)^{\card{V}}\chi^{\zeta_1}(P)(-n) = \sum_{c:V\to[n]} \card{\text{\textit{vertices of $P_c$}}}.
        \end{equation*}
\end{corollary}

\begin{proof}
    From Theorem~\ref{rec_pol_inv_gp} we have that $(-1)^{\card{V}}\chi^{\zeta_1}(P)(-n)$ is the number of compatible pairs
    of points and coloring with $[n]$. Since the points compatible with a coloring $c$ are by definition the points in $P_c$,
    formula~\eqref{cor_poly_gp} follows.
\end{proof}


\section{Polynomial invariants and reciprocity theorems on the Hopf monoid of hypergraphs}\label{poly_inv_hyper}
As in the previous section, we give here an explicit combinatorial interpretation of the polynomial invariant and its reciprocity 
theorem over the Hopf monoid of hypergraphs defined in \cite{AA}. We expand more on this Hopf monoid and give two proofs
of both the combinatorial interpretation of $\chi$ over positive and non negative integers. One of these proofs is self contained and
the other one uses the results of the previous section.

A {\em hypergraph over $V$} is a multiset $h$ of non empty parts of $V$ called 
{\em edges}. In this context the elements of $V$ are called vertices of $h$. The species of hypergraphs is
denoted by $\HG$. Note that two hypergraphs over different sets can never be equal, e.g $\{\{1,2,3\},\{2,3,4\}\}\in\HG[[4]]$ 
is not the same as $\{\{1,2,3\},\{2,3,4\}\}\in\HG[[4]\cup\{a,b\}]$. This is illustrated in the following figure.

\begin{equation}\begin{split}\label{fig_same_edges}
    \begin{tikzpicture}[Centering,scale=1]
        \node[NodeHyper](1)at(-2,0){$1$};
        \node[NodeHyper](a)at(-1,0){$2$};
        \node[NodeHyper](b)at(0,0){$3$};
        \node[NodeHyper](c)at(0,-1){$4$};
        \draw[EdgeHyper](-1,0.5)--(0,0.5);
        \draw[EdgeHyper](0,0.5)arc(90:0:0.5);
        \draw[EdgeHyper](0.5,0)--(0.5,-1);
        \draw[EdgeHyper](-0.5,-1)arc(-180:0:0.5);
        \draw[EdgeHyper](-0.5,-1)arc(0:90:0.5);
        \draw[EdgeHyper](-1,0.5)arc(90:270:0.5);
        \draw[EdgeHyper](-2,0.25)--(0,0.25);
        \draw[EdgeHyper](-2,0.25)arc(90:270:0.25);
        \draw[EdgeHyper](-2,-.25)--(0,-.25);
        \draw[EdgeHyper](0,0.25)arc(90:-90:0.25);
    \end{tikzpicture} \qquad\qquad\qquad\qquad\qquad
    \begin{tikzpicture}[Centering,scale=1]
        \node[NodeHyper](1)at(-2,0){$1$};
        \node[NodeHyper](a)at(-1,0){$2$};
        \node[NodeHyper](b)at(0,0){$3$};
        \node[NodeHyper](c)at(0,-1){$4$};
        \node[NodeHyper](2)at(-2,-1){$a$};
        \node[NodeHyper](3)at(-1,-1){$b$};
        \draw[EdgeHyper](-1,0.5)--(0,0.5);
        \draw[EdgeHyper](0,0.5)arc(90:0:0.5);
        \draw[EdgeHyper](0.5,0)--(0.5,-1);
        \draw[EdgeHyper](-0.5,-1)arc(-180:0:0.5);
        \draw[EdgeHyper](-0.5,-1)arc(0:90:0.5);
        \draw[EdgeHyper](-1,0.5)arc(90:270:0.5);
        \draw[EdgeHyper](-2,0.25)--(0,0.25);
        \draw[EdgeHyper](-2,0.25)arc(90:270:0.25);
        \draw[EdgeHyper](-2,-.25)--(0,-.25);
        \draw[EdgeHyper](0,0.25)arc(90:-90:0.25);
    \end{tikzpicture} \\
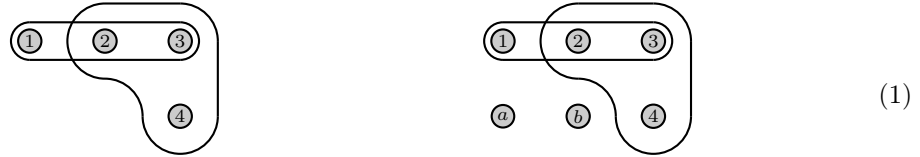
    \text{Two hypergraphs with same edges but over different sets.}
\end{split}\end{equation}


\begin{remark}
    When working only with hypergraphs, we will prefer two different ways of drawing the edges depending on the context. The same edge
    $\set{a,b,c}$ can thus be drawn in the two following ways:
    
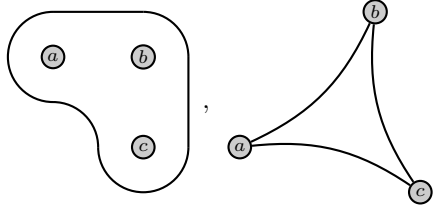
\begin{equation*}
      \begin{tikzpicture}[Centering,scale=1.2]
            \node[NodeHyper](a)at(-1,0){$a$};
            \node[NodeHyper](b)at(0,0){$b$};
            \node[NodeHyper](c)at(0,-1){$c$};
            \draw[EdgeHyper](-1,0.5)--(0,0.5);
            \draw[EdgeHyper](0,0.5)arc(90:0:0.5);
            \draw[EdgeHyper](0.5,0)--(0.5,-1);
            \draw[EdgeHyper](-0.5,-1)arc(-180:0:0.5);
            \draw[EdgeHyper](-0.5,-1)arc(0:90:0.5);
            \draw[EdgeHyper](-1,0.5)arc(90:270:0.5);
        \end{tikzpicture}
        \enspace , \enspace
        \begin{tikzpicture}[Centering,scale=1.2]
            \node[NodeHyper](a)at(-2,0.5){$a$};
            \node[NodeHyper](b)at(-0.5,2){$b$};
            \node[NodeHyper](c)at(0,0){$c$};
            \draw[EdgeHyper](a)edge[bend left=20](c);
            \draw[EdgeHyper](a)edge[bend right=20](b);
            \draw[EdgeHyper](b)edge[bend right=20](c);
        \end{tikzpicture}
        \enspace .
    \end{equation*}
\end{remark}

The species $\HG$ has a Hopf monoid structure with the product the disjoint union and the co-product given by $\Delta_{V_1,V_2}(h) = h|_{V_1} \otimes h/_{V_1}$
where
\begin{equation*}
    h|_V = \set{e\in h\, |\, e\subseteq V} \qquad \text{ and }\qquad h/_V= \set{e\cap V^c\, |\, e \nsubseteq V}.
\end{equation*}

\begin{example}\label{ex_co-product_hypergraphs}
    For $V=[5]$, $V_1=\{1,2,5\}$ and $V_2=\{3,4\}$, we have the following co-product:
    
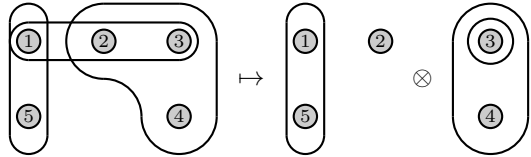
\begin{equation*}
        \begin{tikzpicture}[Centering,scale=1]
            \node[NodeHyper](1)at(-2,0){$1$};
            \node[NodeHyper](a)at(-1,0){$2$};
            \node[NodeHyper](b)at(0,0){$3$};
            \node[NodeHyper](c)at(0,-1){$4$};
            \node[NodeHyper](5)at(-2,-1){$5$};
            \draw[EdgeHyper](-1,0.5)--(0,0.5);
            \draw[EdgeHyper](0,0.5)arc(90:0:0.5);
            \draw[EdgeHyper](0.5,0)--(0.5,-1);
            \draw[EdgeHyper](-0.5,-1)arc(-180:0:0.5);
            \draw[EdgeHyper](-0.5,-1)arc(0:90:0.5);
            \draw[EdgeHyper](-1,0.5)arc(90:270:0.5);
            \draw[EdgeHyper](-2,0.25)--(0,0.25);
            \draw[EdgeHyper](-2,0.25)arc(90:270:0.25);
            \draw[EdgeHyper](-2,-.25)--(0,-0.25);
            \draw[EdgeHyper](0,0.25)arc(90:-90:0.25);
            \draw[EdgeHyper](-2.25,0.25)--(-2.25,-1.25);
            \draw[EdgeHyper](-1.75,0.25)--(-1.75,-1.25);
            \draw[EdgeHyper](-1.75,0.25)arc(0:180:0.25);
            \draw[EdgeHyper](-2.25,-1.25)arc(-180:0:0.25);
        \end{tikzpicture}
        \enspace \mapsto \enspace
        \begin{tikzpicture}[Centering,scale=1]
            \node[NodeHyper](1)at(-2,0){$1$};
            \node[NodeHyper](a)at(-1,0){$2$};
            \node[NodeHyper](5)at(-2,-1){$5$};
            \draw[EdgeHyper](-2.25,0.25)--(-2.25,-1.25);
            \draw[EdgeHyper](-1.75,0.25)--(-1.75,-1.25);
            \draw[EdgeHyper](-1.75,0.25)arc(0:180:0.25);
            \draw[EdgeHyper](-2.25,-1.25)arc(-180:0:0.25);
        \end{tikzpicture} 
        \enspace \otimes \enspace
        \begin{tikzpicture}[Centering,scale=1]
            \node[NodeHyper](b)at(0,0){$3$};
            \node[NodeHyper](c)at(0,-1){$4$};
            \draw[EdgeHyper](b)circle(0.3);
            \draw[EdgeHyper](-0.5,0)--(-0.5,-1);
            \draw[EdgeHyper](0.5,0)--(0.5,-1);
            \draw[EdgeHyper](0.5,0)arc(0:180:0.5);
            \draw[EdgeHyper](-0.5,-1)arc(-180:0:0.5);
        \end{tikzpicture}
    \end{equation*}
\end{example}

The Hopf structure defined here is different from the one defined and studied in \cite{HypB} but there are similar
results in both this paper and \cite{HypB}. This is due to the fact that the notion of acyclic orientations plays a
huge role when working on hypergraphs. Up to some minor subtleties, our notion of acyclic orientations is the same as the
one in \cite{HypB}: we trade the notion of flats for acyclic orientations where all the vertices of an edge can be a target. 
We preferred to define them as maps over hypergraphs with certain properties while they are defined
as compositions of edges which induce directed acyclic graphs. We will point out the common results when they appear.

While the proofs using results of the previous section are rather short,
the self contained proofs are more involved. 
Subsection~\ref{faulhaber_pol} presents some preliminary results. The combinatorial interpretations of $\chi(n)$ and $\chi(-n)$ 
and their self contained proofs are in subsection~\ref{self_contained_pol_hyp} and subsection~\ref{self_contained_rec_th}. 
The proof using the previous section is presented in subsection~\ref{alt_proof}.


\subsection{Generalized Faulhaber's polynomials}\label{faulhaber_pol}
As stated in Theorem~\ref{th_polynomial_invariant}, for any character $\zeta$ and any hypergraph $h$, 
$\chi^{\zeta}(h)(n)$ is a polynomial in $n$. The objects defined here are useful tools to show and 
exploit this polynomial dependency and were introduced in \cite{my1}.

Let $p = (p_1, p_2,\dots, p_t)$ be a finite sequence of positive integers. We define the {\em generalized 
Faulhaber polynomial over $p$}, $F_p$, the function over the integers given by, for $n\in\N$:
\begin{equation*}\label{faulhaber_polynomial}
    F_p(n) = \sum_{0\leq k_1 <\dots < k_t\leq n-1} k_1^{p_1}\cdots k_t^{p_t}.
\end{equation*}
Note that if $t>n$, then $F_p(n) = \sum_{\emptyset}\dots=0$.

As their names suggest, these functions are polynomials.
\begin{proposition}[Proposition 13 in \cite{my1}]\label{prop_faulhaber_polynomial}
    Let $p_1, p_2,\dots, p_t$ be integers and define $d_k = \sum_{i=1}^k p_i + k$ for $1\leq k \leq t$. 
    Then $F_{p_1,\dots, p_t}$ is a polynomial of degree $d_t$ whose constant coefficient is null and whose 
    $(d_t - i)$-th (for $i < d_t$) coefficient is given by
    \begin{equation*}
        \sum_{j_{t-1} = 0}^{\min(j_t, d_{t-1}-1)}
        \sum_{j_{t-2}}^{\min(j_{t-1}, d_{t-2}-1)}\dots
        \sum_{j_1 = 0}^{\min(j_2,d_1 - 1)}
        \prod_{k=1}^t \binom{d_k - j_{k-1}}{j_k - j_{k-1}}\frac{B_{j_k - j_{k-1}}}{d_k - j_{k-1}},
    \end{equation*}
    where $j_t = i$ and $j_0 = 0$, and the $B_j$ numbers are the Bernoulli numbers with the convention
     $B_1 = -1/2$.
\end{proposition}

\begin{remark}
    These polynomials also generalize Stirling numbers of the first kind: denote by $F_{1^k}(n)$ the 
    generalized Faulhaber polynomial associated to the sequence of size $k$ with all elements equal to $1$.
    $F_{1^k}(n) = \sum_{0\leq k_1 <\dots < k_t\leq n-1} \prod k_i$ is indeed the absolute value of the 
    coefficient of $x^{n-k}$ in $x(x-1)\cdots(x-n+1)$ and hence $F_{1^k}(n)=s(n,k)$.
\end{remark}

\begin{lemma}[Lemma 21 in \cite{my1}]\label{lemma_faulhaber_polynomials}
    Let $p$ be a sequence of positive integers of length $t$. We then have 
    \begin{equation*}\label{faulhaber_reciprocity}
        F_{p}(-n) = (-1)^{d_t}\sum_{p\prec q} F_q(n+1),
    \end{equation*}
    where $d_t$ is defined in the same way as in Proposition~\ref{prop_faulhaber_polynomial}.
\end{lemma}


\subsection{Chromatic polynomials of hypergraphs}\label{self_contained_pol_hyp}
Before stating our results on $\chi$ in Theorem~\ref{th_chromatic_polynomial}, recall that a
coloring of $V$ with $[n]$ is a map from $V$ to $[n]$ and that there is a canonical bijection between decompositions and
colorings.

\begin{definition}\label{def_colorings}
    Let $h$ be a hypergraph over $V$ and $c$ be a coloring. For $v\in e\in h$, we say that $v$ is a {\em maximal 
    vertex of $e$ (for $c$)} if $v$ is of maximal color in $e$ and we call the {\em maximal color of $e$ (for $c$)} 
    the color of a maximal vertex of $e$. We say that a vertex $v$ is a {\em maximal vertex (for $c$)} if it is a 
    maximal vertex of an edge and that a color is a {\em maximal color (for $c$)} if it is the maximal color of an edge.

    If $W\subseteq V$ is a subset of vertices, the {\em order of appearance of $W$ (for $V$)} is the composition 
    $\forget(c\cap W)$ where $c\cap W = (c_1\cap W,\dots, c_{l(c)}\cap W)$ and the map $\forget$ sends any decomposition to 
    the composition obtained by dropping the empty parts. 
\end{definition}


When working with colorings, by a slight abuse of notation, for $W$ a set of vertices of the same color for $c$,
we will denote by $c(W)$ their color. This extends $c$ to a map from monochromatic sets of vertices to $[n]$.

\begin{example}\label{ex_colorings}
    We represent a hypergraph along a coloring on $V=\{a,b,c,d,e,f\}$ with $\set{{\color{Green}1},
    {\color{Blue}2},{\color{Magenta}3},{\color{Red}4}}$:
    \begin{equation*}
        \begin{tikzpicture}[Centering,scale=1]
            \node[NodeHyper, draw=Blue, fill =Blue!20](e)at(-2,0){$e$};
            \node[NodeHyper, draw=Magenta, fill =Magenta!20](c)at(-1,0){$c$};
            \node[NodeHyper, draw=Green, fill =Green!20](f)at(0,0){$f$};
            \node[NodeHyper, draw=Magenta, fill =Magenta!20](d)at(0,-1){$d$};
            \node[NodeHyper, draw=Red, fill =Red!20](b)at(-1,-1){$b$};
            \node[NodeHyper, draw=Red, fill =Red!20](a)at(-2,-1){$a$};
            \node(e1)at(-2.5,-0.5){$e_1$};
            \node(e2)at(-0.5,-0.4){$e_2$};
            \node(e3)at(0.7,-0.5){$e_3$};
            \node(e4)at(-1,-1.5){$e_4$};
            \draw[EdgeHyper](-1,0.5)--(0,0.5);
            \draw[EdgeHyper](0,0.5)arc(90:0:0.5);
            \draw[EdgeHyper](0.5,0)--(0.5,-1);
            \draw[EdgeHyper](-0.5,-1)arc(-180:0:0.5);
            \draw[EdgeHyper](-0.5,-1)arc(0:90:0.5);
            \draw[EdgeHyper](-1,0.5)arc(90:270:0.5);
            \draw[EdgeHyper](-2,0.25)--(0,0.25);
            \draw[EdgeHyper](-2,0.25)arc(90:270:0.25);
            \draw[EdgeHyper](-2,-.25)--(0,-0.25);
            \draw[EdgeHyper](0,0.25)arc(90:-90:0.25);
            \draw[EdgeHyper](-2.25,0.25)--(-2.25,-1.25);
            \draw[EdgeHyper](-1.75,0.25)--(-1.75,-1.25);
            \draw[EdgeHyper](-1.75,0.25)arc(0:180:0.25);
            \draw[EdgeHyper](-2.25,-1.25)arc(-180:0:0.25);
            \draw[EdgeHyper](-1,-0.75)--(0,-0.75);
            \draw[EdgeHyper](-1,-1.25)--(0,-1.25);
            \draw[EdgeHyper](0,-0.75)arc(90:-90:0.25);
            \draw[EdgeHyper](-1,-0.75)arc(90:270:0.25);
        \end{tikzpicture}
    \end{equation*}
    The maximal vertex of $e_1$ is $a$ and the maximal vertices of $e_3$ are $c$ and $d$. The maximal color of 
    $e_2$ is {\color{Magenta}{3}}. The order of appearance of $\{a,c,d,e\}$ is $(\{e\},\{c,d\},\{a\})$.
\end{example}

\begin{definition}\label{orientations}
Let $h$ be a hypergraph over $V$.
An {\em admissible orientation} $f$ of $h$ is a map from $h$ to $\p(V)\setminus\set{\emptyset}$ such 
that $f(e)\subseteq e$ for every edge $e$. It is an {\em acyclic orientation} if there is no sequence of distinct edges 
$e_1,\dots, e_k$ such that: $f(e_i)\cap f_s(e_{i+1})\noemp$ or 
$\emptyset\subsetneq f(e_i)\cap e_{i+1} \subsetneq f(e_{i+1})$ for $1\leq i < k$ and $f(e_k)\cap f_s(e_1) \noemp$, where $f_s(e)=f(e)\setminus f(e)$.
\end{definition}

\begin{example}\label{ex_orientations}
    We represent here a hypergraph with two different orientations. One is cyclic and the other is discrete and acyclic.
    \begin{equation*}
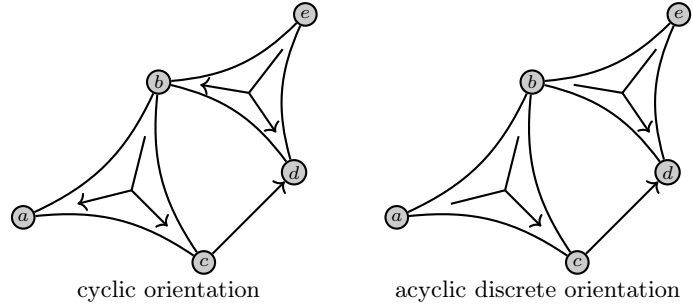
        \begin{tabular}{cc}
            \begin{tikzpicture}[Centering,scale=1.2]
                \node[NodeHyper](a)at(-2,0.5){$a$};
                \node[NodeHyper](b)at(-0.5,2){$b$};
                \node[NodeHyper](c)at(0,0){$c$};
                \node[NodeHyper](d)at(1,1){$d$};
                \node[NodeHyper](e)at(1.125,2.75){$e$};
                \draw[EdgeHyper](a)edge[bend left=20](c);
                \draw[EdgeHyper](a)edge[bend right=20](b);
                \draw[EdgeHyper](b)edge[bend right=20](c);
                \draw[EdgeHyper,->](c)--(d);
                \draw[EdgeHyper](b)edge[bend right=20](e);
                \draw[EdgeHyper](b)edge[bend left=20](d);
                \draw[EdgeHyper](e)edge[bend right=20](d);
                \draw[EdgeHyper](-0.8,0.8)--(-0.65,1.4);
                \draw[EdgeHyper,->](-0.8,0.8)--(-1.4,0.65);
                \draw[EdgeHyper,->](-0.8,0.8)--(-0.4,0.4);
                \draw[EdgeHyper](0.5,1.88)--(0.875,2.365);
                \draw[EdgeHyper,->](0.5,1.88)--(-0.04,1.97);
                \draw[EdgeHyper,->](0.5,1.88)--(0.8,1.44);
            \end{tikzpicture} \enspace & \enspace
            \begin{tikzpicture}[Centering,scale=1.2]
                \node[NodeHyper](a)at(-2,0.5){$a$};
                \node[NodeHyper](b)at(-0.5,2){$b$};
                \node[NodeHyper](c)at(0,0){$c$};
                \node[NodeHyper](d)at(1,1){$d$};
                \node[NodeHyper](e)at(1.125,2.75){$e$};
                \draw[EdgeHyper](a)edge[bend left=20](c);
                \draw[EdgeHyper](a)edge[bend right=20](b);
                \draw[EdgeHyper](b)edge[bend right=20](c);
                \draw[EdgeHyper,->](c)--(d);
                \draw[EdgeHyper](b)edge[bend right=20](e);
                \draw[EdgeHyper](b)edge[bend left=20](d);
                \draw[EdgeHyper](e)edge[bend right=20](d);
                \draw[EdgeHyper](-0.8,0.8)--(-0.65,1.4);
                \draw[EdgeHyper](-0.8,0.8)--(-1.4,0.65);
                \draw[EdgeHyper,->](-0.8,0.8)--(-0.4,0.4);
                \draw[EdgeHyper](0.5,1.88)--(0.875,2.365);
                \draw[EdgeHyper](0.5,1.88)--(-0.04,1.97);
                \draw[EdgeHyper,->](0.5,1.88)--(0.8,1.44);
            \end{tikzpicture} \\
            \text{\small{cyclic orientation}}\enspace & \enspace
            \text{\small{acyclic discrete orientation}}
        \end{tabular}
    \end{equation*}
\end{example}

Given a hypergraph $h$ and $f$ an admissible orientation of $h$, the image of $h$ by $f$, $f(h)$ is also a
hypergraph: $f(h)=\set{f(e)\,|\, e\in h}$.

\begin{definition}\label{def_orientations}
    Let $h$ be a hypergraph over $V$ and $f$ an admissible orientation of $h$. 
    \begin{itemize}
        \item For $c$ coloring of $V$, $c$ and $f$ are said to be {\em compatible} if for every edge $e$ the 
        elements of $f(e)$ are maximal in $e$ for $c$. They are said to be {\em strictly compatible} if $f(e)$ 
        is exactly the set of maximal elements of $e$ for $c$. We denote by $\colorings{h,f,n}$ the set of colorings
        of $V$ with $[n]$ compatible with $f$ and by $\scolorings{h,f,n}$ the set of those with strict compatibility.
        \item For $\zeta$ a character of $\HG$, $f$ is said to be a {\em $\zeta$-orientation} of $h$ if $\zeta(f(h))\not=0$.
    \end{itemize}
\end{definition}

\begin{example}\label{ex_compatible}
    \begin{itemize}
        \item The coloring given in Example~\ref{ex_colorings} has three compatible acyclic orientations: both send $e_1$ 
        on $a$, $e_2$ on $c$ and $e_4$ on $b$, but $e_3$ can be either sent over $c$, $d$ or $\set{c,d}$. Among these 
        three only the last one is strictly compatible.

        Here is an example of an acyclic orientation of a hypergraph with a compatible coloring and a strictly compatible 
        coloring with $\set{{\color{Green}1}, {\color{Magenta}2}, {\color{Red}3}}$:
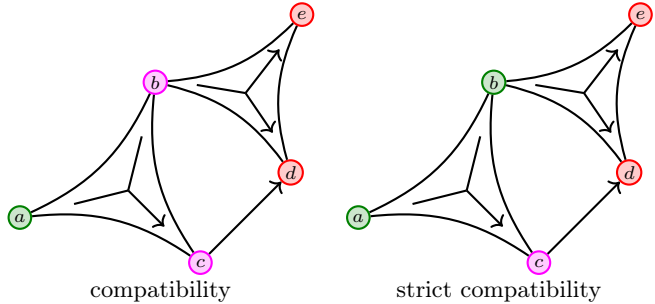
        \begin{equation*}
            \begin{tabular}{cc}
                \begin{tikzpicture}[Centering,scale=1.2]
                    \node[NodeHyper, draw=Green, fill=Green!20](a)at(-2,0.5){$a$};
                    \node[NodeHyper, draw=Magenta, fill=Magenta!20](b)at(-0.5,2){$b$};
                    \node[NodeHyper, draw=Magenta, fill=Magenta!20](c)at(0,0){$c$};
                    \node[NodeHyper, draw=Red, fill=Red!20](d)at(1,1){$d$};
                    \node[NodeHyper, draw=Red, fill=Red!20](e)at(1.125,2.75){$e$};
                    \draw[EdgeHyper](a)edge[bend left=20](c);
                    \draw[EdgeHyper](a)edge[bend right=20](b);
                    \draw[EdgeHyper](b)edge[bend right=20](c);
                    \draw[EdgeHyper,->](c)--(d);
                    \draw[EdgeHyper](b)edge[bend right=20](e);
                    \draw[EdgeHyper](b)edge[bend left=20](d);
                    \draw[EdgeHyper](e)edge[bend right=20](d);
                    \draw[EdgeHyper](-0.8,0.8)--(-0.65,1.4);
                    \draw[EdgeHyper](-0.8,0.8)--(-1.4,0.65);
                    \draw[EdgeHyper,->](-0.8,0.8)--(-0.4,0.4);
                    \draw[EdgeHyper,->](0.5,1.88)--(0.875,2.365);
                    \draw[EdgeHyper](0.5,1.88)--(-0.04,1.97);
                    \draw[EdgeHyper,->](0.5,1.88)--(0.8,1.44);
                \end{tikzpicture}
                &
                \begin{tikzpicture}[Centering,scale=1.2]
                    \node[NodeHyper, draw=Green, fill=Green!20](a)at(-2,0.5){$a$};
                    \node[NodeHyper, draw=Green, fill=Green!20](b)at(-0.5,2){$b$};
                    \node[NodeHyper, draw=Magenta, fill=Magenta!20](c)at(0,0){$c$};
                    \node[NodeHyper, draw=Red, fill=Red!20](d)at(1,1){$d$};
                    \node[NodeHyper, draw=Red, fill=Red!20](e)at(1.125,2.75){$e$};
                    \draw[EdgeHyper](a)edge[bend left=20](c);
                    \draw[EdgeHyper](a)edge[bend right=20](b);
                    \draw[EdgeHyper](b)edge[bend right=20](c);
                    \draw[EdgeHyper,->](c)--(d);
                    \draw[EdgeHyper](b)edge[bend right=20](e);
                    \draw[EdgeHyper](b)edge[bend left=20](d);
                    \draw[EdgeHyper](e)edge[bend right=20](d);
                    \draw[EdgeHyper](-0.8,0.8)--(-0.65,1.4);
                    \draw[EdgeHyper](-0.8,0.8)--(-1.4,0.65);
                    \draw[EdgeHyper,->](-0.8,0.8)--(-0.4,0.4);
                    \draw[EdgeHyper,->](0.5,1.88)--(0.875,2.365);
                    \draw[EdgeHyper](0.5,1.88)--(-0.04,1.97);
                    \draw[EdgeHyper,->](0.5,1.88)--(0.8,1.44);
                \end{tikzpicture} \\
                \text{\small{compatibility}}  &
                \text{\small{strict compatibility}}
            \end{tabular}
        \end{equation*}
        \item Let $\zeta$ be the characteristic function of hypergraphs which connected components are either of size 3 or an isolated vertex. 
        The preceding orientation is not a $\zeta$-orientation since the image of the hypergraph by this orientation is the hypergraph
        
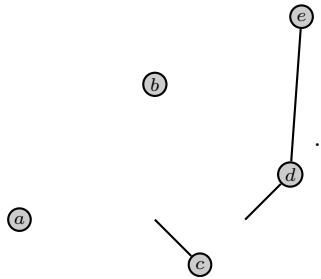
\begin{equation*}
            \begin{tikzpicture}[Centering,scale=1.2]
                \node[NodeHyper](a)at(-2,0.5){$a$};
                \node[NodeHyper](b)at(-0.5,2){$b$};
                \node[NodeHyper](c)at(0,0){$c$};
                \node[NodeHyper](d)at(1,1){$d$};
                \node[NodeHyper](e)at(1.125,2.75){$e$};
                \draw[EdgeHyper](c)--(-0.5,0.5);
                \draw[EdgeHyper](d)--(0.5,0.5);
                \draw[EdgeHyper](d)--(e);
            \end{tikzpicture}.
        \end{equation*}
        Here is an example of a $\zeta$-orientation and the image of the hypergraph by it:
        
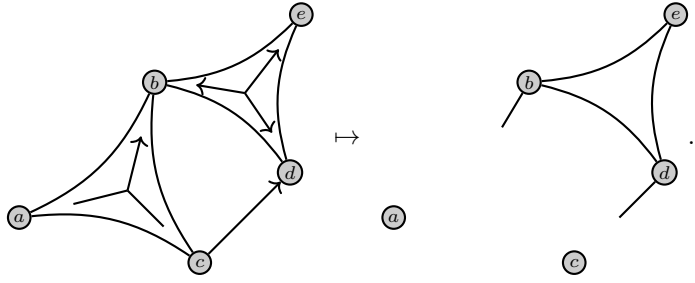
\begin{equation*}
            \begin{tikzpicture}[Centering,scale=1.2]
                \node[NodeHyper](a)at(-2,0.5){$a$};
                \node[NodeHyper](b)at(-0.5,2){$b$};
                \node[NodeHyper](c)at(0,0){$c$};
                \node[NodeHyper](d)at(1,1){$d$};
                \node[NodeHyper](e)at(1.125,2.75){$e$};
                \draw[EdgeHyper](a)edge[bend left=20](c);
                \draw[EdgeHyper](a)edge[bend right=20](b);
                \draw[EdgeHyper](b)edge[bend right=20](c);
                \draw[EdgeHyper,->](c)--(d);
                \draw[EdgeHyper](b)edge[bend right=20](e);
                \draw[EdgeHyper](b)edge[bend left=20](d);
                \draw[EdgeHyper](e)edge[bend right=20](d);
                \draw[EdgeHyper,->](-0.8,0.8)--(-0.65,1.4);
                \draw[EdgeHyper](-0.8,0.8)--(-1.4,0.65);
                \draw[EdgeHyper](-0.8,0.8)--(-0.4,0.4);
                \draw[EdgeHyper,->](0.5,1.88)--(0.875,2.365);
                \draw[EdgeHyper,->](0.5,1.88)--(-0.04,1.97);
                \draw[EdgeHyper,->](0.5,1.88)--(0.8,1.44);
            \end{tikzpicture}
            \enspace \mapsto \enspace
            \begin{tikzpicture}[Centering,scale=1.2]
                \node[NodeHyper](a)at(-2,0.5){$a$};
                \node[NodeHyper](b)at(-0.5,2){$b$};
                \node[NodeHyper](c)at(0,0){$c$};
                \node[NodeHyper](d)at(1,1){$d$};
                \node[NodeHyper](e)at(1.125,2.75){$e$};
                \draw[EdgeHyper](d)--(0.5,0.5);
                \draw[EdgeHyper](b)--(-0.8,1.5);
                \draw[EdgeHyper](b)edge[bend right=20](e);
                \draw[EdgeHyper](b)edge[bend left=20](d);
                \draw[EdgeHyper](e)edge[bend right=20](d);
            \end{tikzpicture}.
        \end{equation*}
    \end{itemize}
\end{example}

Remark that with these definitions, given a coloring $c$ and a hypergraph $h$, there is a unique orientation of $h$
strictly compatible with $c$, which is defined by $f(e)=\set{v\in e\,|\, c(v)=\max(c(e))}$. Furthermore,
this orientation is necessarily acyclic. Indeed, suppose $e_1,\dots, e_k$ is a directed cycle in $f$. then 
$f(e_k)\cap f_s(e_1)\noemp$ implies that $c(f(e_k)) < c(f(e_1))$ and for $1\leq i < k$ either $f(e_i)\cap f_s(e_{i+1})$,
which would imply $c(f(e_i)) < c(f(e_{i+1}))$ or $f(e_i)\cap f(e_{i+1})\noemp$ which would imply $c(f(e_i)) = c(f(e_{i+1}))$.
This then gives us $c(f(e_i))<c(f(e_1))$ which is absurd. We will denote by $\max_c$ this orientation.

With the same kind of reasoning, any coloring compatible with a cyclic orientation must be monochromatic on directed
cycles.
The study of $\scolorings{h,f,n}$ and $\colorings{h,f,n}$ is hence more interesting when $f$ is acyclic and we have an 
expression of $\scolorings{f,h,n}$ in terms of generalized Faulhaber polynomials. Recall that for every hypergraph $h$
over $V$, we have a decomposition of $V$ in the set of isolated vertices and vertices in an edge: $V=\I(h)\sqcup \NI(h)$.

\begin{proposition}\label{prop_colorings}
    Let $h$ be a hypergraph over $V$, $f$ be an acyclic orientation of $h$. Define $P_{h,f}$ and $P'_{h,f}$ as
    the set of compositions $P\vDash f(h)$ such that for $e$ and $e'$ two edges of $h$,
    \begin{itemize} 
        \item if $f(e)\cap f(e')\noemp$ then $P(f(e))=P(f(e'))$,
        \item else if $f(e)\cap f_s(e')\noemp$ then $P(f(e))< P(f(e'))$ for $P\in P_{h,f}$ and $P(f(e))\leq P(f(e'))$ for $P\in P'_{h,f}$.
    \end{itemize}
    We then have
    \begin{equation*}\begin{split}
        \card{\scolorings{h,f,n}} &= n^{\card{\I(h)}}\sum_{P\in P_{h,f}} F_{p_1,\dots, p_{l(P)}}(n)\text{ and } \\
        \card{\colorings{h,f,n}} &= n^{\card{\I(h)}}\sum_{P\in P'_{h,f}} F_{p_1,\dots, p_{l(P)}}(n+1),
    \end{split}\end{equation*}
    where for every composition $P$, $p_i =\card{\widetilde{P_i}}$ and 
    $\widetilde{P_i} = \NI(f(h))^c\cap \left(\bigcup_{e\in f^{-1}(P_i)}e\right) \bigcap_{j<i} \widetilde{P_j}^c$.
\end{proposition}

\begin{example}
    To prove the formula for $\card{\scolorings{h,f,n}}$ we will show that there is a bijection between the set of strictly compatible colorings and the set 
    \begin{equation*}
        \bigsqcup_{P\in P_{h,f}}\bigsqcup_{0\leq k_1 < k_2 < \dots < k_{l(P)}\leq n-1} \prod_{1\leq i \leq l(P)}[k_i]^{\widetilde{P_i}}.
    \end{equation*} 
    We give here an example of how this bijection will work.
    Let $H$, $F$ and $C$ be the hypergraph, the acyclic orientation and the strictly compatible coloring with $\set{{\color{Cyan}1}, {\color{Green}2},{\color{Blue}3},{\color{Magenta}4},{\color{Orange}5},{\color{Red}6}}$ represented here:
    \begin{equation*}
        \begin{tikzpicture}[Centering,scale=1.2]
            \node[NodeHyper, minimum size = 5mm, draw=Red, fill=Red!20](a)at(0,0){$a$};
            \node[NodeHyper, minimum size = 5mm, draw=Red, fill=Red!20](b)at(2,0){$b$};
            \node[NodeHyper, minimum size = 5mm, draw=Blue, fill=Blue!20](c)at(1,2){$c$};
            \node[NodeHyper, minimum size = 5mm, draw=Green, fill=Green!20](d)at(-1,2){$d$};
            \node[NodeHyper, minimum size = 5mm, draw=Green, fill=Green!20](e)at(-2,0){$e$};
            \node[NodeHyper, minimum size = 5mm, draw=Cyan, fill=Cyan!20](f)at(-3,2){$f$};
            \node[NodeHyper, minimum size = 5mm, draw=Blue, fill=Blue!20](g)at(-3,3){$g$};
            \node[NodeHyper, minimum size = 5mm, draw=Magenta, fill=Magenta!20](h)at(-3,5){$h$};
            \node[NodeHyper, minimum size = 5mm, draw=Orange, fill=Orange!20](i)at(-1,4){$i$};
            \node[NodeHyper, minimum size = 5mm, draw=Orange, fill=Orange!20](j)at(1,4){$j$};

            \draw[EdgeHyper](a)edge[bend left=20](b);
            \draw[EdgeHyper](a)edge[bend right=20](c);
            \draw[EdgeHyper](c)edge[bend right=20](b);
            \draw[EdgeHyper](a)edge[bend left=20](d);
            \draw[EdgeHyper](a)edge[bend right=20](e);
            \draw[EdgeHyper](e)edge[bend right=20](d);
            \draw[ArcHyper](f)--(d);
            \draw[EdgeHyper](d)edge[bend left=20](c);
            \draw[EdgeHyper](d)edge[bend right=20](i);
            \draw[EdgeHyper](i)edge[bend right=20](j);
            \draw[EdgeHyper](j)edge[bend right=20](c);
            \draw[EdgeHyper](h)edge[bend right=20](i);
            \draw[EdgeHyper](h)edge[bend left=20](g);
            \draw[EdgeHyper](g)edge[bend left=20](i);

            \draw[ArcHyper](1,0.7)--(1.5,0.3);
            \draw[ArcHyper](1,0.7)--(0.5,0.3);
            \draw[ArcHyper](-1,0.7)--(-.5,0.3);
            \draw[ArcHyper](-2.3,4)--(-1.7,4);
            \draw[ArcHyper](0,3)--(-0.5,3.5);
            \draw[ArcHyper](0,3)--(.5,3.5);
        \end{tikzpicture}.
    \end{equation*}
    The image of $C$ is obtained in the following way: 
    \begin{itemize}
        \item $P = (d,ij,ab)\in P_{H,F}$ is the relative order of the maximal vertices for $C$,
        \item $k_1={\color{Cyan}1}$ ,$k_2={\color{Magenta}4}$ and,$k_3={\color{Orange}5}$ are the colors of the vertices in $P$ shifted by $-1$, the vertices in $P_i$
        being of color $k_i+1$,
        \item we then have $\widetilde{P_1}=\set{f}$, $\widetilde{P_2}=\set{c,g,h}$ and $\widetilde{P_3}=\set{e}$ and the image of $C$ is the triplet
        $(f_1, f_2, f_3)$ defined by
        \begin{align*}
            f_1 : \set{f} &\to \set{{\color{Cyan}1}} \\
            f&\mapsto{\color{Cyan}1} \nonumber\\
            f_2 : \set{c,g,h}&\to\set{{\color{Cyan}1}, {\color{Green}2},{\color{Blue}3},{\color{Magenta}4}} \\
            c&\mapsto {\color{Blue}3} \nonumber\\
            g&\mapsto{\color{Blue}3} \nonumber\\
            h&\mapsto{\color{Magenta}4} \nonumber\\
            f_3 : \set{e}&\to\set{{\color{Cyan}1}, {\color{Green}2},{\color{Blue}3},{\color{Magenta}4},{\color{Orange}5}} \\
            e&\mapsto {\color{Green}2}.\nonumber 
        \end{align*}
    \end{itemize}
\end{example}

\begin{proof}
    We first prove the formula for $\card{\scolorings{h,f,n}}$. The term $n^{\card{\I(h)}}$ in the formula is trivially 
    obtained and we hence consider that $h$ has no isolated vertices.

    Informally, the formula can be obtained by the following reasoning. To choose a coloring strictly compatible with $f$, 
    one can proceed in the following way:
    \begin{enumerate}
        \item choose a part ordering of the sets of maximal vertices: $P\in P_{h,f}$,
        \item choose the color of these vertices: $k_1+1,\dots,k_{l(P)}+1$,
        \item choose the colors of the yet uncolored vertices which are in the same edge than vertices of minimal color in 
        $f(h)$: $k_1^{|\widetilde{P_1}|}$; then those in the same edge than a vertices of second minimal color in $f(h)$: 
        $k_2^{|\widetilde{P_2}|}$), etc.
    \end{enumerate}
    More formally, we show that there exists a bijection between the set of strictly compatible colorings and the set 
    \begin{equation*}
        \bigsqcup_{P\in P_{h,f}}\bigsqcup_{0\leq k_1 < k_2 < \dots < k_{l(P)}\leq n-1} \prod_{1\leq i \leq l(P)}[k_i]^{\widetilde{P_i}},
    \end{equation*} 
    where $[k_i]^{\widetilde{P_i}}$ is the set of maps from $\widetilde{P_i}$ to $[k_i]$. 

    For any subset $A$ of $\N$, we denote by $\bij_A$ the unique increasing bijection from $A$ to $[\card{A}]$. Let $c$ be a 
    strictly compatible coloring, that is to say, $f=\max_c$. 
    We begin by constructing its image by the announced bijection. Recall that $c$ extends to a map with domain monochromatic 
    set of vertices and hence $c(f(h))$ is the set of maximal colors of $c$. Define:
    \begin{itemize}
        \item the part ordering of the set $f(h)$ of maximal vertices: $P = \bij_{c(f(h))} \circ c$. Here again we consider $c$ as 
        a map from monochromatic set of vertices to $[n]$.
        \item The colors of the maximal vertices: $k_i = c(P_i) - 1$ for $1\leq i \leq l(P)$.
        \item The remaining vertices: $\widetilde{P_i} = \NI(f(h))^c\cap\left(\bigcup_{e\in f^{-1}(P_i)}e\right)\bigcap_{j<i} \widetilde{P_j}^c$ 
        for $1\leq i\leq l(P)$.
    \end{itemize}
    The part ordering $P$ is in $P_{h,f}$ because the strict compatibility of $c$ with $f$ implies that if $f(e)\cap f(e')\noemp$ then $c(f(e)) = 
    c(f(e)\cap f(e'))=c(f(e'))$ and else if $f(e)\cap f_s(e')\noemp$ then $c(f(e)) < c(f(e'))$ and $\bij_{c(f(h))}$ is increasing by 
    definition. The sequence $k_1,\dots, k_{l(P)}$ is increasing since 
    $c(P_i) = c (\bij_{c(f(h))} \circ c)^{-1}(i) = \bij_{c(f(h))}^{-1}(i)$ and $\bij_{c(f(h))}^{-1}$ is increasing by definition. The vertices 
    in $\widetilde{P_i}$ are the vertices which share an edge with a set of vertices in $P_i$ ($\bigcup_{e\in f^{-1}(P_i)}e$) but which are not 
    maximal ($\NI(f(h))^c$). Since the set of vertices in $P_i$ are of color $c(P_i)=k_i+1$, the colors of the vertices in $\widetilde{P_i}$ are necessarily 
    in $[k_i]$. Hence the map $c_{|\widetilde{P_i}}$, that is to say $c$ restricted to $\widetilde{P_i}$ is indeed of co-domain $[k_i]$. 
    We then define the image of $c$ as the tuple $(c_{|\widetilde{P_1}},\dots,c_{|\widetilde{P_{l(P)}}})$. 

    Let us now consider the other direction of the bijection. Let be a partition $P\in P_{h,f}$, a sequence of integers $1\leq k_1 < \dots < k_{l(P)}$ and 
    $(c_1,\dots c_{l(P)})\in \prod_{1\leq i \leq l(P)}[k_i]^{\widetilde{P_i}}$. Define $c: V \to [n]$ by $c_{|\widetilde{P_i}} = c_i$ and 
    $c_{|f(h)}(v) = \bij_{\set{k_1+1,\dots,k_{l(P)}+1}}^{-1}(P(f(e)))$ for $v\in f(e)$; this is well defined since if
    $v\in f(e)\cap f(e')$ then $f(e)\cap f(e')\noemp\Rightarrow P(f(e))=P(f(e'))$. This map $c$ has indeed domain $V$ since 
    $(\widetilde{P_1},\dots ,\widetilde{P_{l(P)}},\NI(f(h)))$ is a partition of $V$.

    Let us show that $c$ is a coloring strictly compatible with $f$. If $v,v'$ are two vertices in $f(e)$ then by definition $c(v)=c(v')$. 
    Let now be $v\in e\setminus f(e)$.
    \begin{itemize}
        \item If $v\in f(e')$ then necessarily $f(e)\cap f(e') = \emptyset$ because otherwise $e',e,e'$ would be a directed cycle. Indeed, $v$ is
        an exit of $e'$ and an entry of $e$ and $f(e)\cap f(e')\noemp$ would imply that $e'$ share an exit with $e$, but not all (not $v$). 
        Hence, by definition of $P_{h,f}$, $P(f(e'))<P(f(e)$) and so $c(v) < c(f(e))$. 
        \item If $v\not\in f(h)$ then $v\in \widetilde{P_i}$ with $i\leq P(f(e))$ and we do have the desired inequality: 
        $c(v) = c_i(v) \leq k_i < k_i + 1 \leq k_{P(f(e))} + 1 = c(f(e))$.
    \end{itemize}
    We conclude this first part of the proof by remarking that the two previous constructions are inverse functions.

    The proof for the formula of $\card{\colorings{h,f,n}}$ is the same except that we now show a bijection with
    \begin{equation*}
        \bigsqcup_{P\in P_{h,f}'}\bigsqcup_{0\leq k_1 \leq k_2 \leq \dots \leq k_{l(P)}\leq n} \prod_{1\leq i \leq l(P)}[k_i]^{\widetilde{P_i}},
    \end{equation*} 
    and that to do so the only change with what precedes is choosing $k_i=c(P_i)$ instead of $k_i=c(P_i)-1$.
\end{proof}

We now state the first theorem of this section:
\begin{theorem}\label{th_chromatic_polynomial}
    Let $\zeta$ be a character of $\HG$, $V$ be a finite set and $h\in\HG[V]$ a hypergraph. Then
    \begin{equation*}
        \chi^{\zeta}_V(h)(n) = \sum_{f\in\acyclic_h} \zeta(f(h))\card{\scolorings{h,f,n}}.
    \end{equation*}
    In particular, if $\zeta$ is a characteristic function, then $\chi^{\zeta}(h)(n)$ is the number of strictly 
    compatible pairs of acyclic $\zeta$-orientations of $h$ and colorings with $[n]$. In this case we call $\chi^{\zeta}$
    the {\em $\zeta$-chromatic polynomial of $h$}.
\end{theorem}

\begin{proof}
    Recall the definition of $\chi^{\zeta}$ from \eqref{def_polynomial_invariant}: 
    \begin{equation*}
        \chi^{\zeta}_V(h)(n)=\sum_{D\in\Dcomp[V], l(D) =n} \zeta_V\circ\mu_D\circ\Delta_D(h).
    \end{equation*} 
    The bijection $b$ defined at \eqref{bij_decomposition_colorings} provides a bijection between the decomposition of size $n$
    of $V$ and the colorings of $V$ with $[n]$. To prove our theorem, we just need to show that for $D$ a coloring of $V$, 
    $\max_D(h)=\mu_D\circ\Delta_D$ since we would then have
    \begin{equation*}\begin{split}
        \chi^{\zeta}_V(h)(n)&=\sum_{D\in\Dcomp[V], l(D) =n} \zeta_V\circ\mu_D\circ\Delta_D(h) \\
        &=\sum_{\text{$D$ coloring of $f$ with $[n]$}} \zeta_V(\max_D(h)) \\
        &=\sum_{f\in\acyclic_h} \zeta_V(f(h))\card{\scolorings{h,f,n}}.
    \end{split}\end{equation*}
    We prove this by induction over $n$. Let $D=(D_1,D_2)$ be a decomposition of $V$ of size $2$. Then 
    $\mu_D\circ\Delta_D(h)= h|_{D_1}\sqcup h/_{D_2} = \set{e\in h\,|\, e\subseteq D_1}\sqcup\set{e\cap D_2\,|\, e\subsetneq D_1}$.
    We then need to show that $\max_D(e)=e$ when $e\subseteq D_1$ and $\max_D(e)=e\cap D_2$ else. If $e\subseteq D_1$, then that means
    that all the vertices in $e$ are maximal in $e$ of color 1, and hence $\max_D(e)=e$. If $e\subsetneq D_1$, then the maximal vertices
    of $e$ are of the one of color 2 \textit{i.e.}  $\max_D(e)= e\cap D_2$. This concludes the case $n=2$.

    Suppose now that this statement is true for $n-1$ and let $D$ be a decomposition of $V$ of size $n$. Denote by $W$ the set $V\setminus D_n$
    and by $D_{|W}=(D_1,\dots,D_{n-1})$ the restriction of $D$ to $W$ (as a map). By induction we have that
    \begin{equation*}\begin{split}
        \mu_D\circ\Delta_D(h) &= \mu_{W,D_n}\circ\mu_{D_{|W}}\otimes\id\circ\Delta_{D_{|W}}\otimes\id\circ\Delta_{W,D_n}(h) \\
        &= \mu_{W,D_n}\circ \max_{D_{|W}}\otimes\id\circ\Delta_{W,D_n}(h) \\
        &= \mu_{W,D_n}\circ \max_{D_{|W}}\otimes\id\circ h|_W\otimes h/_W \\
        &= \max_{D_{|W}}(h|_W)\sqcup h/_W,
    \end{split}\end{equation*}
    and hence we must show that $\max_D(e) = \max_{D_{|W}}(e)$ when $e\subseteq W$ and $\max_D(e)= e\cap D_n$ otherwise. The first assertion
    is straightforward by definition of the restriction and the second assertion is proven in an analogous way that the case $n=2$.
    This concludes this proof.
\end{proof}

It is clear that connected hypergraphs play the role of prime elements in $\HG$. This means that every hypergraph can be uniquely written
as a product of connected elements, up to the order. When defining a character, we hence only need to define it on the connected
hypergraphs. In particular, we only define characteristic functions by specifying the connected hypergraphs with value 1. 

\begin{example}
    Let $\zeta_{e3}$ be the characteristic function of edges of size three and isolated vertices. Then $\chi^{\zeta_{e3}}$ counts the colorings such that there
    is exactly three maximal vertices by edge and edges do not share maximal vertices. 

    Let now $\zeta_3$ be the characteristic function of connected hypergraphs over three vertices and isolated vertices. Then $\chi^{\zeta_3}$ counts the number of 
    colorings such that reducing each edge to its maximal vertices gives us a hypergraph where connected components are either either of size 3 or an isolated vertex.

    Of the three following colorings with $\set{{\color{Green}1},{\color{Blue}2},{\color{Red}3}}$, the first two are counted by $\chi^{\zeta_{e3}}$ but not the third. 
    None of them are counted by $\chi^{\zeta_3}$.
    
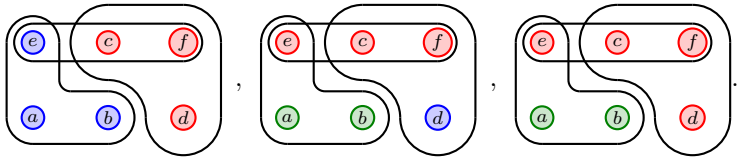
\begin{equation*}
        \begin{tikzpicture}[Centering,scale=1]
            \node[NodeHyper, draw=Blue, fill =Blue!20](e)at(-2,0){$e$};
            \node[NodeHyper, draw=Red, fill =Red!20](c)at(-1,0){$c$};
            \node[NodeHyper, draw=Red, fill =Red!20](f)at(0,0){$f$};
            \node[NodeHyper, draw=Red, fill =Red!20](d)at(0,-1){$d$};
            \node[NodeHyper, draw=Blue, fill =Blue!20](b)at(-1,-1){$b$};
            \node[NodeHyper, draw=Blue, fill =Blue!20](a)at(-2,-1){$a$};
            \draw[EdgeHyper](-1,0.5)--(0,0.5);
            \draw[EdgeHyper](0,0.5)arc(90:0:0.5);
            \draw[EdgeHyper](0.5,0)--(0.5,-1);
            \draw[EdgeHyper](-0.5,-1)arc(-180:0:0.5);
            \draw[EdgeHyper](-0.5,-1)arc(0:90:0.5);
            \draw[EdgeHyper](-1,0.5)arc(90:270:0.5);
            \draw[EdgeHyper](-2,0.25)--(0,0.25);
            \draw[EdgeHyper](-2,0.25)arc(90:270:0.25);
            \draw[EdgeHyper](-2,-.25)--(0,-0.25);
            \draw[EdgeHyper](0,0.25)arc(90:-90:0.25);

            \draw[EdgeHyper](e)pic{carc=0:180:0.35};

            \draw[EdgeHyper](-2.35,0)--(-2.35,-1);
            \draw[EdgeHyper](-1.65,0)--(-1.65,-0.5);

            \draw[EdgeHyper](-1.5,-0.5)pic{carc=180:270:0.15};

            \draw[EdgeHyper](a)pic{carc=180:270:0.35};

            \draw[EdgeHyper](-1.5,-0.65)--(-1,-0.65);
            \draw[EdgeHyper](-2,-1.35)--(-1,-1.35);
            \draw[EdgeHyper](b)pic{carc=90:-90:0.35};
        \end{tikzpicture}
        \enspace, \enspace
        \begin{tikzpicture}[Centering,scale=1]
            \node[NodeHyper, draw=Red, fill =Red!20](e)at(-2,0){$e$};
            \node[NodeHyper, draw=Red, fill =Red!20](c)at(-1,0){$c$};
            \node[NodeHyper, draw=Red, fill =Red!20](f)at(0,0){$f$};
            \node[NodeHyper, draw=Blue, fill =Blue!20](d)at(0,-1){$d$};
            \node[NodeHyper, draw=Green, fill =Green!20](b)at(-1,-1){$b$};
            \node[NodeHyper, draw=Green, fill =Green!20](a)at(-2,-1){$a$};
            \draw[EdgeHyper](-1,0.5)--(0,0.5);
            \draw[EdgeHyper](0,0.5)arc(90:0:0.5);
            \draw[EdgeHyper](0.5,0)--(0.5,-1);
            \draw[EdgeHyper](-0.5,-1)arc(-180:0:0.5);
            \draw[EdgeHyper](-0.5,-1)arc(0:90:0.5);
            \draw[EdgeHyper](-1,0.5)arc(90:270:0.5);
            \draw[EdgeHyper](-2,0.25)--(0,0.25);
            \draw[EdgeHyper](-2,0.25)arc(90:270:0.25);
            \draw[EdgeHyper](-2,-.25)--(0,-0.25);
            \draw[EdgeHyper](0,0.25)arc(90:-90:0.25);

            \draw[EdgeHyper](e)pic{carc=0:180:0.35};

            \draw[EdgeHyper](-2.35,0)--(-2.35,-1);
            \draw[EdgeHyper](-1.65,0)--(-1.65,-0.5);

            \draw[EdgeHyper](-1.5,-0.5)pic{carc=180:270:0.15};

            \draw[EdgeHyper](a)pic{carc=180:270:0.35};

            \draw[EdgeHyper](-1.5,-0.65)--(-1,-0.65);
            \draw[EdgeHyper](-2,-1.35)--(-1,-1.35);
            \draw[EdgeHyper](b)pic{carc=90:-90:0.35};
        \end{tikzpicture}
        \enspace , \enspace
        \begin{tikzpicture}[Centering,scale=1]
            \node[NodeHyper, draw=Red, fill =Red!20](e)at(-2,0){$e$};
            \node[NodeHyper, draw=Red, fill =Red!20](c)at(-1,0){$c$};
            \node[NodeHyper, draw=Red, fill =Red!20](f)at(0,0){$f$};
            \node[NodeHyper, draw=Red, fill =Red!20](d)at(0,-1){$d$};
            \node[NodeHyper, draw=Green, fill =Green!20](b)at(-1,-1){$b$};
            \node[NodeHyper, draw=Green, fill =Green!20](a)at(-2,-1){$a$};
            \draw[EdgeHyper](-1,0.5)--(0,0.5);
            \draw[EdgeHyper](0,0.5)arc(90:0:0.5);
            \draw[EdgeHyper](0.5,0)--(0.5,-1);
            \draw[EdgeHyper](-0.5,-1)arc(-180:0:0.5);
            \draw[EdgeHyper](-0.5,-1)arc(0:90:0.5);
            \draw[EdgeHyper](-1,0.5)arc(90:270:0.5);
            \draw[EdgeHyper](-2,0.25)--(0,0.25);
            \draw[EdgeHyper](-2,0.25)arc(90:270:0.25);
            \draw[EdgeHyper](-2,-.25)--(0,-0.25);
            \draw[EdgeHyper](0,0.25)arc(90:-90:0.25);

            \draw[EdgeHyper](e)pic{carc=0:180:0.35};

            \draw[EdgeHyper](-2.35,0)--(-2.35,-1);
            \draw[EdgeHyper](-1.65,0)--(-1.65,-0.5);

            \draw[EdgeHyper](-1.5,-0.5)pic{carc=180:270:0.15};

            \draw[EdgeHyper](a)pic{carc=180:270:0.35};

            \draw[EdgeHyper](-1.5,-0.65)--(-1,-0.65);
            \draw[EdgeHyper](-2,-1.35)--(-1,-1.35);
            \draw[EdgeHyper](b)pic{carc=90:-90:0.35};
        \end{tikzpicture}.
    \end{equation*}
\end{example}

Recall that the basic character $\zeta_1$ is the characteristic function of discrete hypergraphs. We have a particular interpretation 
of $\chi^{\zeta_1}$.

\begin{corollary}[Theorem 18 in \cite{my1}]\label{cor_chrom_pol}
    Let $V$ be a finite set and $h\in\HG[V]$ a hypergraph. Then $\chi^{\zeta_1}(h)(n)$ is the number of colorings of $V$ such that every 
    edge of $h$ has only one maximal vertex.
\end{corollary}

\begin{proof}
    This proof is straightforward: the $\zeta_1$-orientations are exactly the discrete orientations and the colorings compatible with 
    such orientations are colorings where each edge has exactly one maximal vertex.
\end{proof}

\begin{example}
    The coloring given in Example~\ref{ex_colorings} is not counted in $\chi^{\zeta_1}_V(h)(4)$ since $e_3$ has two maximal vertices. 
    However by changing the color of $d$ to {\color{Blue}2} we do obtain a coloring where every edge has only one maximal vertex.

    Let $g$ be the hypergraph $\{\{1,2,3\},\{2,3,4\}\}\in \HG[[4]]$ represented in Figure~\ref{fig_same_edges}. We then have 
    $\chi^{\zeta_1}_{[4]}(h)(n)= n^4 - \frac{8}{3}n^3+\frac{5}{2}n^2-\frac{5}{6}n$ and we verify that, for example, $\chi_{[4]}(h)(2)=3$.
\end{example}



\subsection{Reciprocity theorem}\label{self_contained_rec_th}
We now give the reciprocity theorem which gives us an expression of $\chi^{\zeta}$ over negative integers as well as a combinatorial interpretation when possible.
A character $\zeta$ of $\HG$ is {\em odd} if $\zeta(h)=0$ for every $h$ with a connected component with an even number of vertices.
This can also be expressed by stating that the only connected hypergraphs on which $\zeta$ is not null have odd number of vertices.

Denote by $\cc(h)$ the number of connected components of $h$.
\begin{theorem}\label{th_reciprocity_hypergraphs}
    Let $\zeta$ be a character of $\HG$, $V$ be a finite set and $h\in\HG[V]$ a hypergraph then
    \begin{equation*}
        \chi^{\zeta}_V(h)(-n) = \sum_{f\in\acyclic_h} (-1)^{\cc(f(h))}\zeta(f(h))\card{\colorings{h,f,n}}.
    \end{equation*}
    Furthermore, if $\zeta$ is an odd characteristic function then $(-1)^{\card{V}}\chi^{\zeta}(h)(-n)$ 
    is the number of compatible pairs of acyclic $\zeta$-orientations of $h$ and colorings with $[n]$. In particular we have in this
    case that $(-1)^{\card{V}}\chi^{\zeta}(h)(-1)$ is the number of acyclic $\zeta$-orientations of $h$.
\end{theorem}

\begin{corollary}[Theorem 24 in \cite{my1}]\label{cor_rec_th}
    Let $V$ be a finite set and $h\in\HG[V]$ a hypergraph. Then $(-1)^{\card{V}}\chi^{\zeta_1}(h)(-n)$ is the number of compatible
    discrete acyclic orientation of $h$ and colorings with $[n]$. In particular, we have now that $(-1)^{\card{V}}\chi^{\zeta_1}(h)(-1)$ is the number
    of discrete acyclic orientations of $h$.
\end{corollary}

\begin{example}
    For any any hypergraph $h$ over $V$ and any odd character $\zeta$ of $\HG$, we have $\chi_V(h)(n) \leq (-1)^{|V|}\chi_V(h)(-n)$. 
    This comes from the fact that any strictly compatible pair is a compatible pair. This is observed for $\zeta=\zeta_1$ and 
    $h=\{\{1,2,3\},\{2,3,4\}\}\in \HG[[4]]$:
    \begin{equation*}
        \chi_{[4]}(h)(n)= n^4 - \frac{8}{3}n^3+\frac{5}{2}n^2-\frac{5}{6}n < n^4 + \frac{8}{3}n^3+\frac{5}{2}n^2+\frac{5}{6}n = (-1)^4\chi_{[4]}(h)(-n).
    \end{equation*}
    We also verify that $h$ does have $\chi_{[4]}(h)(-1) = 7$ acyclic discrete orientations ($3\times 3$ orientations minus the two cyclic orientations).
\end{example}

As announced at the beginning of this section, we give here a self-contained proof which uses the previous results on generalized 
Faulhaber's polynomials and on compatible colorings. Except for the defitions of some objects, this proof is essentialy the same 
than the one given in \cite{my1} for the case of the basic character and we will also need the following lemma.

Recall from subsection~\ref{decomp_op} the classical definitions over compositions: product, shuffle product and refinement.

\begin{lemma}[Lemma 23 in \cite{my1}]\label{lemma_composition_sum}
    Let $V$ be a set and $P\vDash V$ a composition of $V$. We have the identity:
    \begin{equation*}
        \sum_{Q\prec P} (-1)^{l(Q)} = (-1)^{\card{V}}.
    \end{equation*}

    Furthermore, let $g$ be a directed acyclic graph on $V$ and consider the {\em constrained set} \linebreak[4] 
    $C(g,P)=\{Q\prec P\,|\, \forall (v,v')\in g,Q(v)< Q(v')\}$. We have the more general identity:
    \begin{equation*}
        \sum_{Q\in C(g,P)} (-1)^{l(Q)} = \left\{\begin{array}{cl}
        0  & \text{if there exists $(v,v')\in g$ such that $P(v')<P(v)$},    \\ 
        (-1)^{\card{V}}  & \text{if not.}\end{array}\right.
    \end{equation*}
\end{lemma}

While we interpreted Lemma~\ref{lemma_composition_sum} as a result on graphs and partitions, it can also be seen as a result on posets 
and linear extensions.

We can now give our first proof to Theorem~\ref{th_reciprocity_hypergraphs}.

\begin{proof}[Proof of Theorem~\ref{th_reciprocity_hypergraphs}.]
    From Theorem~\ref{th_chromatic_polynomial}, Proposition~\ref{prop_colorings} and Lemma~\ref{lemma_faulhaber_polynomials}, we have that: 
    \begin{equation*}
        \chi^{\zeta}_V(h)(-n) = (-n)^{\card{\I(h)}}\sum_{f\in \acyclic_h}\zeta(f(h))\sum_{P\in P_{h,f}}(-1)^{\sum_{i=1}^{l(P)}p_i + l(P)}\sum_{(p_1, \dots, p_{l(P)}) \prec q}F_q(n+1).
    \end{equation*}
    Let $P$ be a composition. We then have:
    \begin{itemize}
        \item $\sum_{i=1}^{l(P)}p_i = \card{\NI(h)} - \card{\NI(f(h))}$, since $(\widetilde{P_1},\dots \widetilde{P_{l(P)}}, \NI(f(h)))$ 
        is a partition of $\NI(h)$.
        \item The map
        \begin{equation*}\begin{split}
            \phi :\set{Q\vDash f(h)\,|\, P\prec Q} &\to \set{q\vDash(\card{\NI(h)} - \card{\NI(f(h))})\,|\, (p_1,\dots, p_{l(P)}) \prec q}\\
            Q &\mapsto (\card{\widetilde{Q_1}},\dots, \card{\widetilde{Q_{l(Q)}}})
        \end{split}\end{equation*}
        is a bijection ($\widetilde{Q}_i$ is defined in the same way that $\widetilde{P_i}$ in Proposition~\ref{prop_colorings}). Indeed, the two sets 
        have same cardinality $\sum_{k=0}^{l(P)-1}\binom{l(P)}{k}$: in both cases we choose the number $k+1$ of elements of the composition and then which 
        consecutive elements of the composition to merge: $\binom{l(P)}{k}$. Furthermore the map $\phi$ is a surjection, since the 
        composition $(q_1,\dots, q_k)$ with $q_i = \sum_{j = j_i}^{k_i} p_j$ is the image of the composition $Q_1,\dots,Q_k$ with $Q_i = \sqcup_{j = j_i}^{k_i} P_j$.
        This comes from the fact that for any two disjoint sets of sets $A,B$ we have
        $\bigcup_{e\in A\sqcup B}e = \bigcup_{e\in A}e \cup \bigcup_{e\in B}e = \bigcup_{e\in A}e \sqcup \left(\bigcup_{e\in B}e \right)\cap\left(\bigcup_{e\in A}e\right)^c$
        and that the sets $f^{-1}(P_i)$ are pairwise disjoint.
    \end{itemize}
    These two remarks lead to:
    \begin{equation*}\begin{split}
        \chi^{\zeta}_V(h)(-n) &= n^{\card{\I(h)}}\sum_{f\in \acyclic_h}(-1)^{\card{V}-\card{\NI(f(h))}}\zeta(f(h))\sum_{P\in P_{h,f}}(-1)^{l(P)}\sum_{P\prec Q} F_{\phi(Q)}(n+1)\\
        &= n^{\card{\I(h)}}\sum_{f\in \acyclic_h}(-1)^{\card{V}-\card{\NI(f(h))}}\zeta(f(h))\sum_{Q\vDash f(h)}\left(\sum_{\substack{P\prec Q\\P\in P_{h,f}}} (-1)^{l(P)}\right)F_{\phi(Q)}(n+1).
    \end{split}\end{equation*}
    Let $g$ be the graph with vertices the connected components of $f(h)$ and with an oriented edge from a connected component $h_1$ to another
    connected component $h_2$ if there is $e_1\in h_1$ and $e_2\in h_2$ such that $f(e_1)\cap f_s(e_2)\noemp$. Since $f$ is acyclic, $g$ is a directed
    acyclic graph. We can see the compositions in $P_{h,f}$ and $P_{h,f}'$ (Proposition~\ref{prop_colorings}) as compositions over connected components of $f(h)$ 
    since for such compositions, two edges of $f(h)$ with a non empty intersection must be in the same part by definition. Remarking then that with this 
    point of view $\set{ P\prec Q\,|\, P\in P_{h,f}} = C(g,Q)$, Lemma~\ref{lemma_composition_sum} leads to:
    \begin{equation*}\begin{split}
        \chi_V(h)(-n) &= n^{\card{\I(h)}}\sum_{f\in \acyclic_h}(-1)^{\card{V}-\card{\NI(f(h))}}\zeta(f(h))\sum_{\substack{P\vDash f(h)\\ P(v)\leq P(v') \forall(v,v')\in g}} (-1)^{\cc(f(h))-\card{\I(f(h))}}F_{\phi(P)}(n+1) \\
        &= n^{\card{\I(h)}}\sum_{f\in \acyclic_h}(-1)^{\card{V}-\card{\NI(f(h))}-\card{\I(f(h))}+\cc(f(h))}\zeta(f(h))\sum_{P\in P_{h,f}'} F_{p_1,\dots, p_{l(P)}}(n+1) \\
        &= n^{\card{\I(h)}}\sum_{f\in \acyclic_h}(-1)^{\cc(f(h))}\card{\colorings{f,h,n}},
    \end{split}\end{equation*}
    where the last equality is Proposition~\ref{prop_colorings}.

    To complete this proof, note that when $\zeta$ is odd, $\zeta(f(h))\not=0$ implies that $\card{V}-\cc(f(h))$ is even since each connected component $h'$ of $f(h)$
    participate for $V(h')-1$ which is even.
\end{proof}


\subsection{Alternative proof}\label{alt_proof}
Let us now give the second proof. As announced this proof is shorter and we will prove both Theorem~\ref{th_chromatic_polynomial}
and Theorem~\ref{th_reciprocity_hypergraphs} at the same time. 

This proof rest on the existence of a bijection between hypergraphs and hypergraphic polytopes.
Indeed for $h\in\HG[V]$ a hypergraph, 
denote by $\Delta_h$ the Minkowski sum $\sum_{e\in h}\Delta_e$. Then the map $\Delta:h\mapsto\Delta_h$ clearly is a species isomorphism 
from $\HG$ to $\HGP$.
We further extend the similarities of hypergraphs and hypergraphic polytopes with the following lemma.

\begin{lemma}\label{lemma_acyclic_faces}
    Let $V$ be a finite set and $h$ a hypergraph over $V$. Then the faces of $\Delta_h$ are exactly the hypergraphic polytopes 
    $\Delta_{f(h)}$ for $f\in\acyclic_h$. Furthermore, for $f$ an acyclic orientation of $h$, $\scolorings{h,f,n}=\scone{\Delta_h}(\Delta_{f(h)})_n$
    and $\colorings{h,f,n}=\cone{\Delta_h}(\Delta_{f(h)})_n$.
\end{lemma}

\begin{remark}
    This lemma is equivalent to Theorem 2.18 of \cite{HypB}. Our approach and notations being different than in \cite{HypB},
    we preferred to give an alternative proof of this lemma.
\end{remark}

\begin{proof}
    Let be $f\in\acyclic_h$. We first show that $\Delta_{f(h)}$ is indeed a face of $\Delta_h$. Let $c$ be a coloring
    strictly compatible with $f$ \textit{i.e.}  $f=\max_c$. We show that $\Delta_{f(h)}$ 
    is the $c$-maximum face of $\Delta_h$. Let $p$ be a point in $\Delta_h$. Then by definition of $\Delta_h$, $p$ can be written
    as $\sum_{e\in h}\sum_{v\in e}a_{e,v} v$ where for each edge $e$ the $a_{e,v}$ are positive real numbers summing to one:
    $\sum_{v\in e} a_{e,v}=1$. We then have that,
    \begin{equation*}\begin{split}
        c(p) &= \sum_{e\in h}\sum_{v\in e}a_{e,v} c(v)\\
        &=\sum_{e\in h} \left(c(f(e))\sum_{v\in f(e)}a_{e,v}+\sum_{v\in f_s(e)} a_{e,v}c(v)\right),
    \end{split}\end{equation*}
    which is maximum when $a_{e,v} = 0$ for every edge $e$ and every $v\in f_s(e)$. This implies that the $c$-maximum face of $\Delta_h$
    is the set of points of the form $\sum_{e\in h} \sum_{v\in f(e)}a_{e,v}v$ with $\sum_{v\in f(e)} a_{e,v} = 1$. This is exactly
    $\Delta_{f(h)}$.

    Let now $Q$ be a face of $\Delta_h$. For $P$ a polytope and $y$ a direction, the $y$-maximum face $P_y$ does not depend on the exact
    values of $y$ but only on the induced order $v<v'$ if $y(v)<y(v')$. Let hence be $c$ a coloring with value in $[n]$ such 
    that $Q=P_c$. Then by what precedes, $Q=\Delta_{\max_c(h)}$.

    The equalities between the sets of compatible colorings and the cones directly follow from the preceding.
\end{proof}

This lemma together with Proposition~\ref{prop_polynomial_invariant} is enough to give the desired proof.

\begin{proof}[Proof of Theorem~\ref{th_reciprocity_hypergraphs}]
    Let $\zeta$ be a character of $\HG$ and define $\zeta'$ the character of $\GP$ defined by $\zeta'(P) = \zeta(h)$ if $P=\Delta_h$ is a
    hypergraphic polytope and $\zeta'(P)=0$ else. Let $h$ be a hypergraph over $V$. Then, applying Proposition~\ref{prop_polynomial_invariant}, 
    Theorem~\ref{th_pol_inv_gp} and finally Lemma~\ref{lemma_acyclic_faces}, we have:
    \begin{equation*}\begin{split}
        \chi^{\HG,\zeta}_V(h)(n) &=\chi^{\GP,\zeta'}_V(\Delta_h)(n) = \sum_{Q\leq \Delta_h} \zeta'(Q)\card{\scone{\Delta_h}(Q)_n} \\
        &= \sum_{f\in\acyclic_h} \zeta(f(h))\card{\scone{\Delta_h}(\Delta_{f(h)})_n} \\
        &= \sum_{f\in\acyclic_h} \zeta(f(h))\card{\scolorings{h,f,n}}.
    \end{split}\end{equation*} 
    The formula over non positive integers is obtained analogously.
\end{proof}

\begin{remark}
    As a corollary from Lemma~\ref{lemma_acyclic_faces} and Theorem~\ref{th_antipode_gp}, we have that the antipode of $\HG$
    is given by the cancellation-free and grouping-free expression
    \begin{equation*}
        \antipode_V(h)=\sum_{f\in\acyclic_h}(-1)^{\cc(f(h))}f(h).
    \end{equation*}
    While this expression express the antipode only in term of acyclic orientations, faces of polytopes have much more
    apparent structure than acyclic orientations and are easier and more intuitive to work with.
\end{remark}


\section{Other Hopf monoids}\label{sec_other_hopf}

In this section we use Theorem~\ref{th_chromatic_polynomial} and Theorem~\ref{th_reciprocity_hypergraphs} to obtain similar results on other 
Hopf monoids, more precisely the Hopf monoids from sections 19 to 24 of \cite{AA}.
The general method used here is to use the fact that these Hopf monoids can be seen as sub-monoids of (most of the times) the Hopf monoid 
of simple hypergraphs, and then present an interpretation of what is an acyclic orientation on these particular Hopf monoids. More precisely, 
Proposition~\ref{prop_polynomial_invariant} tells us that if we have a morphism $\phi:\spe{M}\to\SHG$ then $\chi^{\spe{M},\zeta\circ\phi}=\chi^{\SHG,\zeta}\circ\phi$
for $\zeta$ any character of $\SHG$. Since $\phi$ will be injective in our cases, restricting its co-domain to is image makes it an isomorphism.
We then have $\chi^{\spe{M},\zeta} = \chi^{\SHG,\zeta\circ\phi^{-1}}\circ\phi$ for $\zeta$ any character of $\spe{M}$. What remains is to 
find a combinatorial interpretation of acyclic orientations on the objects of $\spe{M}$. 

\begin{remark}
    Note that this is exactly how we obtained our second proof of the combinatorial interpretation of $\chi^{\HG}$. Here, since we are only working on purely
    combinatorial objects, the morphism $\phi$ is simpler than in the case of $\Delta:\HG\to\HGP$.
\end{remark}

Not only the results of this section generalize a lot of other results, but they are also obtained with a uniform approach. We provide 
details at the beginning of each subsection on the links between our results and already existing ones.


\subsection{Simple hypergraphs}
A {\em simple hypergraph over $V$} is a set $h$ of non empty parts of $V$. $\SHG$ admits a similar Hopf
monoid structure to $\HG$. In fact its structure maps can be defined in the same way
\begin{align*}
    \mu_{V_1,V_2}: \SHG[V_1]\otimes \SHG[V_2] &\to \SHG[V] & \Delta_{V_1,V_2}: \SHG[V] &\to \SHG[V_1]\otimes \SHG[V_2] \\
    h_1\otimes h_2 &\mapsto h_1\sqcup h_2 & h &\mapsto h|_{V_1}\otimes h/_S{V_2}, \nonumber
\end{align*}
where we also have $h|_V = \set{e\in h\, |\, e\subseteq V}$ and $h/_V= \set{e\cap V^c\, |\, e \nsubseteq V}$. The difference with $\HG$ is that here 
we are working with sets instead of multisets. So even if two edges $e_1$ and $e_1$ are such that $e_1\cap V= e_2\cap V=e$, there will only be one edge $e$
in $h/_V$.

As this structure is very similar to the one over hypergraphs, it is of no surprise that the polynomial invariants also have similar expression.

\begin{proposition}\label{prop_simple_hypergraphs}
    Let $\zeta$ be a character of $\SHG$, $V$ be a finite set and $h\in\SHG[V]$ a simple hypergraph. We then have:
    \begin{equation*}\begin{split}
        \chi^{\zeta}_V(h)(n) &= \sum_{f\in\acyclic_h} \zeta(f(h))\card{\scolorings{h,f,n}},\\
        \chi^{\zeta}_V(h)(-n) &= \sum_{f\in\acyclic_h} (-1)^{\cc(f(h))}\zeta(f(h))\card{\colorings{h,f,n}}. \nonumber
    \end{split}\end{equation*}
    If $\zeta$ is a characteristic function, then $\chi^{\zeta}(h)(n)$ is the number of strictly compatible pairs of acyclic $\zeta$-orientations of 
    $h$ and colorings with $[n]$. Furthermore, if $\zeta$ is odd $(-1)^{\card{V}}\chi^{\zeta}_V(h)(-n)$ is the number of compatible ones. In particular,
    $(-1)^{\card{V}}\chi^{\zeta}_V(h)(-1)$ is the number of acyclic $\zeta$-orientations of $h$.
\end{proposition}

\begin{proof}
    Let $\text{mult}$ be the species morphism from $\SHG$ to $\HG$ which sends a simple hypergraph on the same hypergraph. It is the right inverse of
    $\dom$ which sends a hypergraph to its domain. Let $\zeta$ be a character of $\SHG$, and remark that $\dom$ is a morphism of Hopf monoids so 
    $\zeta\circ\dom$ is a character of $\HG$.
    Then Proposition~\ref{prop_polynomial_invariant} with $\phi=\dom$ gives us, for $h$ a simple hypergraph
    \begin{equation*}
        \chi^{\HG,\zeta\circ\dom}(\text{mult}(h))(n) = \chi^{\SHG,\zeta}(\dom(\text{mult}(h)))(n) = \chi^{\SHG,\zeta}(h)(n).
    \end{equation*}
    The result follows, since $\acyclic_{\text{mult}(h)}=\acyclic_h$. Moreover, the same goes with (strictly) compatible colorings.
\end{proof}

\begin{corollary}
    Let $V$ be a finite set and $h\in\SHG[V]$ a hypergraph. Then $\chi^{\zeta_1}(h)(n)$ is the number of colorings of $V$ such that every 
    edge of $h$ has only one maximal vertex and $(-1)^{\card{V}}\chi^{\zeta_1}(h)(-n)$ is the number of compatible
    discrete acyclic orientation of $h$ and colorings with $[n]$. In particular, $(-1)^{\card{V}}\chi^{\zeta_1}(h)(-1)$ is the number
    of discrete acyclic orientations of $h$. 
\end{corollary}


\subsection{Graphs}\label{cor_graphs}
Recall that A {\em graph over $V$} is a hypergraph whose edges are all of cardinality 2. The species $\G$ is not 
stable under the restriction in $\HG$ but it still admits a close Hopf monoid structure. The 
structure maps are given by
\begin{align*}
\mu_{V_1,V_2}: \G[V_1]\otimes\G[V_2] &\to\G[V] & \Delta_{V_1,V_2}:\G[V] &\to\G[V_1]\otimes\G[V_2] \\
g_1\otimes g_2 &\mapsto g_1\sqcup g_2 & g &\mapsto g|_{V_1}\otimes g/_{V_1}, \nonumber
\end{align*}
where $g/_{V_2}$ is defined in the same way as in $\HG$ and $g|_{V_1}=\set{e\in g\,|\, e\subseteq V_1}$ is the contraction of $V_2$ to $g$ as a 
hypergraph \textit{i.e.}  $g|_{V_1} = g/_{V_2}$.

Usually an orientation of a graph is what we call here a discrete orientation. We hence introduce the notion of partial orientation of a graph which 
we think is more intuitive than the notion of admissible orientation in the case of graphs.
A {\em partial orientation} of a graph $g$ is a discrete orientation of a sub-graph $h$ of $g$. The partial orientations of $g$ are in bijection 
with the admissible orientations of $g$ by the map $\kappa$ which sends a partial orientation $f$ on the admissible orientation $\kappa_f$ defined 
by $\kappa_f(e) = f(e)$, if $f(e)$ is defined, and $\kappa_f(e)=e$ else. For $f$ a partial orientation of $g$, we denote by $f(g)$ the sub-graph of 
$g$ formed of the non-oriented edges. One can think of it as if we followed the oriented edges while erasing them behind us. This is the same than 
the graph obtained by leaving aside the edges of size 1 in $\kappa_f(e)$. A partial orientation is {\em acyclic} if it is trivial (no edge is oriented) 
or it is not possible to complete it in order to obtain a cycle. It is equivalent to say that its image by $\kappa$ is acyclic and we will consider 
$\acyclic_g$ as the set of acyclic partial orientations on $g$ in the rest of this subsection. 

A coloring $c$ of $V$ is {\em strictly compatible (resp compatible)} with a partial orientation $f$ of $g$ if $f(e) = \max_c(e)$ (resp $f(e)\subseteq \max_c(e)$)
when $f(e)$ is defined and the rest of the edges are monochromatic, $\max_c(e)=e$.

\begin{example}
    We represent here a cyclic partial orientation and a coloring with $\set{{\color{Blue}1},{\color{Red}2}}$ along its strictly compatible partial 
    orientation.
    \begin{equation*}
        \begin{tabular}{cc}
            \begin{tikzpicture}[Centering,scale=1.2]
                \node[NodeHyper](a)at(0.5,1){$a$};
                \node[NodeHyper](c)at(1.5,1){$c$};
                \node[NodeHyper](b)at(0,0){$b$};
                \node[NodeHyper](d)at(1,0){$d$};
                \draw[EdgeHyper](b)--(a);
                \draw[EdgeHyper,->](b)--(c);
                \draw[EdgeHyper,->](b)--(d);
                \draw[EdgeHyper,->](c)--(a);
                \draw[EdgeHyper](c)--(d);
            \end{tikzpicture}
            &
            \begin{tikzpicture}[Centering,scale=1.2]
                \node[NodeHyper,draw=Red, fill=Red!20](a)at(0.5,1){$a$};
                \node[NodeHyper,draw=Red, fill=Red!20](c)at(1.5,1){$c$};
                \node[NodeHyper,draw=Blue, fill=Blue!20](b)at(0,0){$b$};
                \node[NodeHyper,draw=Red, fill=Red!20](d)at(1,0){$d$};
                \draw[EdgeHyper, ->](b)--(a);
                \draw[EdgeHyper,->](b)--(c);
                \draw[EdgeHyper,->](b)--(d);
                \draw[EdgeHyper](c)--(a);
                \draw[EdgeHyper](c)--(d);
            \end{tikzpicture} \\
            \text{\small{A cyclic partial orientation}}
            &
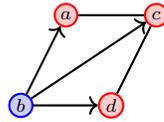
            \text{\small{A coloring and its strictly compatible partial orientation}}
        \end{tabular}
    \end{equation*}
\end{example}

For $\zeta$ a character of $\G$, a $\zeta$-partial orientation of $g$ is a partial orientation $f$ such that $\zeta(f(g))\not=0$. A character $\zeta$ 
is odd if the connected graphs on which it is not null have an odd number of vertices.

\begin{remark}
    In the literature, the preferred notion is that of pairs of flats $F$ and discrete acyclic orientations of the quotient graph $g/F$. A {\em flat} $F$ of a
    graph $g$ is a sub-graph of $g$ of the form $\mu_D\circ\Delta_D(g)=g|_{D_1}\sqcup\dots\sqcup g|_{D_n}$ for $D\vDash g$. The {\em quotient}
    $g/F$ is then the graph obtained by deleting the edges in $F$ and merging all the vertices which shared a connected component in $F$. The bijection
    with acyclic orientations of $g$ is again easy to see: send a pair $(F,a)$ of a flat and a discrete acyclic orientation of $g/F$ on the acyclic orientation
    $f$ defined by $f(e)=a(e)$ if $e\not\in F$ and $f(e)=e$ else.

    We preferred the notion of acyclic partial orientation which is more coherent in our context. All our results over hypergraphs were expressed in terms
    of acyclic orientations and not pairs of flat and acyclic orientation, as is done in \cite{HypB}. Still, we also give our result in term of flats
    for the sake of completeness. For $\zeta$ a character of $\G$, a {\em $\zeta$-flat} of $g$ is a flat on which $\zeta$ is not null.
\end{remark}

\begin{proposition}
    Let $\zeta$ be a character of $\G$, $V$ be a finite set and $g\in\G[V]$ a graph. We then have:
    \begin{equation*}\begin{split}
        \chi^{\zeta}_V(g)(n) &= \sum_{f\in\acyclic_g} \zeta(f(g))\card{\scolorings{g,f,n}} 
        = \sum_{F\in\text{Flats}(g)}\sum_{a\in\acyclic_{g/F}}\zeta(F)\card{\scolorings{g/F,a,n}},\\
        \chi^{\zeta}_V(g)(-n) &= \sum_{f\in\acyclic_g} (-1)^{\cc(f(g))}\zeta(f(g))\card{\colorings{g,f,n}} 
        = \sum_{F\in\text{Flats}(g)}\sum_{a\in\acyclic_{g/F}}\zeta(F)\card{\colorings{g/F,a,n}}.
    \end{split}\end{equation*}
    If $\zeta$ is a characteristic function, then $\chi^{\zeta}(g)(n)$ is the number of strictly compatible pairs of acyclic $\zeta$-partial orientations 
    of $g$ and colorings with $[n]$. Furthermore, if $\zeta$ is odd $(-1)^{\card{V}}\chi^{\zeta}_V(g)(-n)$ the number of compatible ones. In particular,
    $(-1)^{\card{V}}\chi^{\zeta}_V(g)(-1)$ is the number of acyclic $\zeta$-partial orientations of $g$.
\end{proposition}

\begin{proof}
    Let $\HG_{\leq 2}$ be the sub-species of $\HG$ of hypergraphs with edges of size at most 2. The species $\HG_{\leq 2}$ is stable under
    product and co-product and is hence a Hopf sub-monoid. Let $s:\HG_{\leq 2}\to\G$ be the species morphism which forget edges of size 1. 
    Then $s:\HG_{\leq 2}\to\G$ is a morphism of Hopf monoid. Let $\zeta$ be a character of $\G$. Proposition~\ref{prop_polynomial_invariant} gives us, 
    for $g$ a graph,
    \begin{equation*}
        \chi^{\HG_{\leq 2},\zeta\circ s}(g)(n) = \chi^{\G,\zeta}(s(g))(n)=\chi^{\G,\zeta}(g).
    \end{equation*} 
    This concludes the proof.
\end{proof}

A {\em proper coloring} of a graph is a coloring such that no edge has its two vertices of the same color. The {\em chromatic polynomial} 
of a graph $g$ is the polynomial $T_g$ such that $T_g(n)$ is the number of proper colorings with $n$ colors.

\begin{corollary}[Proposition 18.1 in \cite{AA}]
    Let $g$ be a graph. Then $\chi^{\G,\zeta_1}(g)(n)=T_g(n)$.
\end{corollary}

\begin{proof}
    From Corollary~\ref{cor_chrom_pol} we know that $\chi^{\G,\zeta_1}(g)(n)$ is the number of colorings with $[n]$ such that each edge has
    a unique maximal vertex. In the case of a graph, this is equivalent to saying that no edge has its two vertices of the same color, \textit{i.e.}  it
    is a proper coloring.
\end{proof}

In particular, by evaluating $\chi^{\G,\zeta_1}$ on non positive integers, we recover the classical reciprocity theorem of Stanley \cite{stan}.

\begin{remark}
    As was the case for simple hypergraphs and hypergraphs, the polynomial invariants of the Hopf sub-monoid of simple graphs $\SG$ of $\SHG$
    admit the same formulas than the polynomial invariants defined there.
\end{remark}


\subsection{Simplicial complexes}
In \cite{sc} Benedetti, Hallam, and Machacek constructed a combinatorial Hopf algebra of simplicial complexes. In particular they obtained 
results over some polynomial invariant which we generalize in this subsection. 

An {\em abstract simplicial complex}, or simplicial complex, on $V$ is a collection $C$ of subsets of $V$, called {\em faces}, such that any 
non empty subset of a face is a face \textit{i.e.}  $I \in C$ and $\emptyset\subsetneq J \subset I$ implies $J \in C$. We denote by $\SC$ the species 
of simplicial complexes. Proposition 21.1 of \cite{AA}, states that the species $\SC$ of simplicial complexes is a sub-monoid of $\SHG$.

Let us now give a simple lemma which will be useful in this subsection and the next one, see Figure~\ref{fig_lemma_sc} for an example
of what this lemma is about.

\begin{lemma}\label{lemma_inclusion}
    Let $V$ be finite set, $h\in\SHG[V]$ be a simple hypergraph and $f$ an acyclic orientation of $h$. Let and $e$ and $e'$ two edges of $h$ of size 
    at least 2 such that $e'\subset e$. Then if $f(e)\cap e'\noemp$, necessarily $f(e)\cap e' = f(e')$.
\end{lemma}

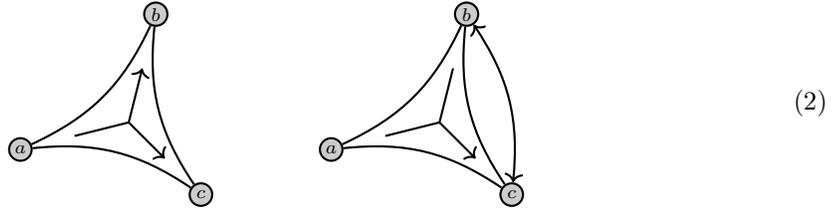
\begin{center}
\begin{equation}\label{fig_lemma_sc}
    \begin{tikzpicture}[Centering,scale=1.2]
        \node[NodeHyper](a)at(-2,0.5){$a$};
        \node[NodeHyper](b)at(-0.5,2){$b$};
        \node[NodeHyper](c)at(0,0){$c$};
        \draw[EdgeHyper](a)edge[bend left=20](c);
        \draw[EdgeHyper](a)edge[bend right=20](b);
        \draw[EdgeHyper](b)edge[bend right=20](c);
        \draw[EdgeHyper,->](-0.8,0.8)--(-0.65,1.4);
        \draw[EdgeHyper](-0.8,0.8)--(-1.4,0.65);
        \draw[EdgeHyper,->](-0.8,0.8)--(-0.4,0.4);
    \end{tikzpicture}
    \qquad\qquad
    \begin{tikzpicture}[Centering,scale=1.2]
        \node[NodeHyper](a)at(-2,0.5){$a$};
        \node[NodeHyper](b)at(-0.5,2){$b$};
        \node[NodeHyper](c)at(0,0){$c$};
        \draw[EdgeHyper](a)edge[bend left=20](c);
        \draw[EdgeHyper](a)edge[bend right=20](b);
        \draw[EdgeHyper](b)edge[bend right=20](c);
        \draw[EdgeHyper,<->](b)edge[bend left=20](c);
        \draw[EdgeHyper](-0.8,0.8)--(-0.65,1.4);
        \draw[EdgeHyper](-0.8,0.8)--(-1.4,0.65);
        \draw[EdgeHyper,->](-0.8,0.8)--(-0.4,0.4);
    \end{tikzpicture}
    \end{equation}
    \small{Two counter examples of Lemma~\ref{lemma_inclusion}. We see that we have a cycle in both cases.}
\end{center}

\begin{proof}
    Let $e$ and $e'$ be two such edges. Suppose there exists $v\in f(e')$ such that $v\not\in f(e)\cap e'$. Then $f(e')\cap f_s(e)\noemp$ and 
    since $f(e)\cap e'\noemp$, we have either $f(e)\cap f_s(e')\noemp$ or the strict inclusions $\emptyset\subsetneq f(e)\cap f(e') \subsetneq f(e')$. 
    This makes the sequence $e',e$ 
    a cycle. Hence $f(e')\subseteq f(e)\cap e'$. Suppose now there exists $v\in f(e)\cap e'$ such that $v\not\in f(e')$. Then similarly to the
    previous case, $f(e)\cap f_s(e')\noemp$ and either $f(e')\cap f_s(e)$ or $\emptyset\subsetneq f(e')\cap f(e)\subsetneq f(e)$. Hence
    $f(e)\cap e' \subseteq f(e')$ and so $f(e')=f(e)\cap e'$.
\end{proof}

The {\em 1-skeleton} of a simplicial complex is the simple graph formed by its faces of cardinality 2. 


\begin{lemma}\label{lemma_sc}
    Let be $V$ a finite set, $C\in\SC[V]$ and $g$ the 1-skeleton of $C$. Then $\acyclic_C\cong\acyclic_g$.
\end{lemma}

\begin{proof}
    The fact that every acyclic orientation of $C$ gives rise to an acyclic orientation of $g$ is clear: if $f\in\acyclic_C$ then a cycle
    in $f_{|g}$ is also a cycle in $f$ and hence it is not possible. This is a bijection because from Lemma~\ref{lemma_inclusion} an orientation
    of a simplicial complex only needs to be defined on its faces of size 2.
\end{proof}

For $C$ a simplicial complex and $f$ an acyclic orientation of its 1-skeleton, we will also denote by $f$ the image of $f$ by this
bijection.

\begin{proposition}
    Let $\zeta$ be a character of $\SC$, $V$ be a finite set, $C\in\SC[V]$ be a simplicial complex and $g$ be the 1-skeleton of $C$. We then have:
    \begin{equation*}\begin{split}
        \chi^{\zeta}_V(C)(n) &= \sum_{f\in\acyclic_g} \zeta(f(C))\card{\scolorings{g,f,n}};\\
        \chi^{\zeta}_V(C)(-n) &= \sum_{f\in\acyclic_g} (-1)^{\cc(f(g))}\zeta(f(C))\card{\colorings{g,f,n}}.
    \end{split}\end{equation*}
    If $\zeta$ is a characteristic function, then $\chi^{\zeta}(C)(n)$ is the number of strictly compatible pairs of acyclic $\zeta$-orientations of 
    $C$ and colorings with $[n]$. Furthermore, if $\zeta$ is odd $(-1)^{\card{V}}\chi^{\zeta}_V(C)(-n)$ is the number of compatible ones. In particular,
    $(-1)^{\card{V}}\chi^{\zeta}_V(C)(-1)$ is the number of acyclic $\zeta$-orientations of $C$.
\end{proposition}

\begin{proof}
    This is just a direct application of Lemma~\ref{lemma_sc}. We just observe that while we do have $\scolorings{C,f,n}=\scolorings{g,f,n}$,
    $\colorings{C,f,n}=\colorings{g,f,n}$ and $\cc(f(C))=\cc(f(g))$, we do not have $f(C)=f(g)$, $f(g)$ being the 1-skeleton of $f(C)$. Hence we
    can not replace $\zeta(f(C))$ and $\zeta$-orientation of $C$ with $\zeta(f(g))$ and $\zeta$-orientation of $g$.
\end{proof}

We say that a character $\zeta$ of $\SC$ is {\em downward compatible} if $\zeta(C)=\zeta(g)$ for any simplicial complex $C$ and $g$ its 1-skeleton.
The map which add to a graph  all the edges of size 1 is an injective map from $\SG$ to $\SHG$. We consider $\SG$ as a sub-species of $\SHG$
under this map.

\begin{corollary}
    Let $\zeta$ be a downward compatible character of $\SC$, $C\in\SC[V]$ be a simplicial complex and $g$ 
    be the 1-skeleton of $C$. Then $\chi^{\SHG,\zeta}_V(C)(n)=\chi^{\SG,\zeta}_V(g)(n)$.
\end{corollary}

\begin{corollary}[Corollary 28 in \cite{my1}]
    Let $V$ be a finite set, $C\in\SC[V]$ be a simplicial complex and $g$ be the 1-skeleton of $C$. Then $\chi^{\zeta_1}_V(C)$ is the chromatic 
    polynomial of $g$.
\end{corollary}

\subsection{Building sets}
Building sets and graphical building sets have been studied in a Hopf algebraic context by Gruji\'c in \cite{gr} where he gave results in link
with the one obtained here.

Building sets were independently introduced by De Concini and Procesi in \cite{DCP} and by Schmitt in \cite{Schmi}. A {\em building set} on 
$V$ is collection $B$ of subsets of $V$, called {\em connected sets}, such that if $I,J\in B$ and $I\cap J\noemp$ then $I\cup J\in B$ and 
for all $v\in V$, $\set{v}\in B$. We denote by $\BS$ the species of building sets. By Proposition 22.3 of \cite{AA} the species $\BS$ 
of building sets is a sub-monoid of $\SHG$.

\begin{definition}[\cite{BS} \cite{Post09}]\label{def_bfor}
    Let $V$ be a finite set and $B$ a building set on $V$. Let $F$ be a forest of rooted trees whose vertices are labelled by the elements of a 
    partition $\pi$ of $V$. We denote by $\leq$ the relation ``is a descendent of'' over $\pi$ implied by $F$ and we denote by
    $F_{\leq p} = \bigsqcup_{q\leq p}q$ for $p\in \pi$. The forest $F$ is then a {\em $B$-forest} if it satisfies
    the three following conditions.
    \begin{enumerate}
        \item If $r$ is a root of $F$, then $F_{\leq r}$ is the set of vertices of the connected component containing $r$.
        \item For any node $p$, $F_{\leq p}\in B$.
        \item For any $k\geq 2$ and pairwise incomparable nodes $p_1,\dots,p_k$, $\bigcup F_{p_i}\not\in B$.
    \end{enumerate}
    We denote by $\bfor{B}$ the set of $B$-forest of $B$ and call {\em $B$-trees} the $B$-forests of $B$ when $B$ is a connected.
\end{definition}

\begin{lemma}\label{lemma_bfor}
    Let $V$ be a finite set and $B$ a connected building set on $V$. Then the $B$-trees also admit the following inductive definition.
    \begin{itemize}
        \item The unique $B$-tree of the building set on a singleton $\set{v}$ is the rooted tree with only its root $\set{v}$.
        \item If $V$ is not a singleton, let $r$ be a subset of $V$ and denote by $V_1,\dots, V_k$ the maximal connected component
        of $B$ which does not intersect with $r$. For $1\leq i\leq k$ let $B_i$ be the connected building set {\em associated to $V_i$},
        which is defined by $B_i=\set{e\in B\,|\, e\subset V_i}$ and let $T_i$ be a $B_i$-tree. Then the tree rooted on $r$ obtained by doing the 
        disjoint union of the $T_i$ and adding an edge between $r$ and the roots of the $T_i$ is a $B$-tree.
    \end{itemize}
\end{lemma}

\begin{remark}
    This lemma is in a way the same as saying that there is a bijection between $B$-forest and nested sets of $B$ (\cite{BS}, \cite{Post09}).
    The formulation given here is more adapted to what we need in the sequel.
\end{remark}

\begin{proof}
    We begin by showing that the tree defined in such a way are indeed $B$-trees. We do this by induction over $\card{V}$. If $V$ is a
    singleton then it is obvious. Suppose now $V$ is not a singleton and let $r$, $V_1,\dots, V_k$ and $B_1,\dots, B_k$ be as defined in
    the lemma.
    \begin{enumerate}
        \item To show the first item, we prove that $V_1,\dots, V_k$ is a partition of $V\setminus r$ because then, by denoting by $r_i$
        the root of $T_i$,
        \begin{equation*}
            T_{\leq r}=\sqcup_{p\leq r}p=r\sqcup_{p<r}p=r\sqcup_i\sqcup_{p\leq r_i}=r\sqcup_i T_{i,\leq r_i} = r\sqcup V_i = V,
        \end{equation*}
        where the fourth equality is obtained by induction. Suppose $V_i\cap V_j\noemp.$ Then since $B$ is a building set, $V_i\cup V_j\in B$, 
        and since neither $V_i$ nor $V_j$ intersect with $r$, their union does not either. This contradicts the maximality of $V_i$ and $V_j$ and 
        it is hence not possible. Let now be $v\in V\setminus r$ then since $\set{V}\in V$, there exists a maximal edge connected component not 
        intersecting $V$ which contains $v$. 
        \item We already showed the second case in the case of $r$. Let $p\not=r$ be a node of $T$. Then there exists $i$ such that $p$ is a
        node of $T_i$ and hence by induction $T_{\leq p}\in B_i\subset B$.
        \item Let be $l\geq 2$ and $p_1,\dots, p_l$ be pairwise incomparable nodes of $T$. Suppose that $\bigcup p_i\in B$. Then if there
        exists $j$ such that all the $p_i$ are in $T_j$, all the $p_i$ would be subset of $V_j$ and we would have $\bigcup p_i \in B_j$ which
        is not possible by induction. So there is at least two indices $i,j$ such that $p_i\in T_{k_i}$ and $p_j\in T_{k_j}$ (the $p_i$s can not
        be equal to $r$ since $r$ is comparable to every node). Then $\bigcup p_i$ is in $B$ but in no $B_i$ and hence it intersects with $r$ 
        by definition. This is not possible since no $p_i$ intersects with $r$. This shows that $\bigcup p_i\not\in B$
    \end{enumerate}

    Let now $T$ be a $B$-tree, $r$ be its root, $T_1,\dots, T_k$ be its direct sub-trees and $r_1,\dots, r_k$ be their respective roots.
    For $1\leq i\leq k$, let $V_i=T_{\leq r_i}=T_{i,\leq r_i}$ and $B_i$ be the connected building set associated to $V_i$. Then by definition 
    of $B$-trees, for $1\leq i\leq k$, $T_i$ is a $B_i$-tree. We only need to show that the $V_i$ are the maximal connected sets which do not
    intersect $r$ to conclude. Clearly they do not intersect $r$ since they are union of nodes and all nodes are pairwise disjoint. Suppose $W$
    is a connected set not intersecting $r$ and such that $V_i\subsetneq W$ for one of the $V_i$s. Since $V_1,\dots, V_k$ forms a partition of 
    $V\setminus r$, $W$ must intersect with some $V_{j_1},\dots,V_{j_l}$, $j_m\not = i$ for $1\leq \leq l$. Suppose without loss of generality
    that $j_m =m$ and $i=l+1$. Then $W\bigcup_{1\leq j\leq l} V_j\in B$ since $B$ is a building set. But this is not possible since this is also 
    equal to $\bigcup_{1\leq j \leq l+1} V_j=\bigcup_{1\leq j \leq l+1} T_{\leq r_j}$ and the $r_j$ are pairwise incomparable. Hence the $V_i$
    are maximal.
\end{proof}

\begin{example}
    We represented here an induction step of the construction of a $B$-tree. The set $W$ is in red and the connected sets in blue are
    the maximal connected sets not intersecting $W$.
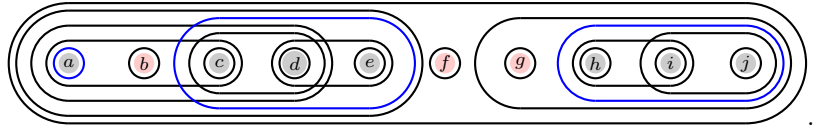
    \begin{equation*}
        \begin{tikzpicture}
            \node[NodeHyper, inner sep =0pt, draw=white](a)at(0,0){$a$};
            \node[NodeHyper, inner sep =0pt, draw=white, fill=Red!20](b)at(1,0){$b$};
            \node[NodeHyper, inner sep =0pt, draw=white](c)at(2,0){$c$};
            \node[NodeHyper, inner sep =0pt, draw=black!20](d)at(3,0){$d$};
            \node[NodeHyper, inner sep =0pt, draw=white](e)at(4,0){$e$};
            \node[NodeHyper, inner sep =0pt, draw=white, fill=Red!20](f)at(5,0){$f$};
            \node[NodeHyper, inner sep =0pt, draw=white, fill=Red!20](g)at(6,0){$g$};
            \node[NodeHyper, inner sep =0pt, draw=white](h)at(7,0){$h$};
            \node[NodeHyper, inner sep =0pt, draw=white](i)at(8,0){$i$};
            \node[NodeHyper, inner sep =0pt, draw=white](j)at(9,0){$j$};
            \draw[EdgeHyper, Blue](a)circle(0.2);
            \draw[EdgeHyper](b)circle(0.2);
            \draw[EdgeHyper](c)circle(0.2);
            \draw[EdgeHyper](d)circle(0.2);
            \draw[EdgeHyper](e)circle(0.2);
            \draw[EdgeHyper](f)circle(0.2);
            \draw[EdgeHyper](g)circle(0.2);
            \draw[EdgeHyper](h)circle(0.2);
            \draw[EdgeHyper](i)circle(0.2);
            \draw[EdgeHyper](j)circle(0.2);

            \draw[EdgeHyper](0, 0.3)--(2,0.3);
            \draw[EdgeHyper](0, -0.3)--(2,-0.3);
            \draw[EdgeHyper](0,0.3)arc(90:270:0.3);
            \draw[EdgeHyper](2,0.3)arc(90:-90:0.3);

            \draw[EdgeHyper](3, 0.3)--(4,0.3);
            \draw[EdgeHyper](3, -0.3)--(4,-0.3);
            \draw[EdgeHyper](3,0.3)arc(90:270:0.3);
            \draw[EdgeHyper](4,0.3)arc(90:-90:0.3);

            \draw[EdgeHyper](7, 0.3)--(8,0.3);
            \draw[EdgeHyper](7, -0.3)--(8,-0.3);
            \draw[EdgeHyper](7,0.3)arc(90:270:0.3);
            \draw[EdgeHyper](8,0.3)arc(90:-90:0.3);

            \draw[EdgeHyper](2, 0.4)--(3,0.4);
            \draw[EdgeHyper](2, -0.4)--(3,-0.4);
            \draw[EdgeHyper](2,0.4)arc(90:270:0.4);
            \draw[EdgeHyper](3,0.4)arc(90:-90:0.4);

            \draw[EdgeHyper](8, 0.4)--(9,0.4);
            \draw[EdgeHyper](8, -0.4)--(9,-0.4);
            \draw[EdgeHyper](8,0.4)arc(90:270:0.4);
            \draw[EdgeHyper](9,0.4)arc(90:-90:0.4);

            \draw[EdgeHyper](0, 0.5)--(3,0.5);
            \draw[EdgeHyper](0, -0.5)--(3,-0.5);
            \draw[EdgeHyper](0,0.5)arc(90:270:0.5);
            \draw[EdgeHyper](3,0.5)arc(90:-90:0.5);

            \draw[EdgeHyper, Blue](7, 0.5)--(9,0.5);
            \draw[EdgeHyper, Blue](7, -0.5)--(9,-0.5);
            \draw[EdgeHyper, Blue](7,0.5)arc(90:270:0.5);
            \draw[EdgeHyper, Blue](9,0.5)arc(90:-90:0.5);

            \draw[EdgeHyper, Blue](2, 0.6)--(4,0.6);
            \draw[EdgeHyper, Blue](2, -0.6)--(4,-0.6);
            \draw[EdgeHyper, Blue](2,0.6)arc(90:270:0.6);
            \draw[EdgeHyper, Blue](4,0.6)arc(90:-90:0.6);

            \draw[EdgeHyper](6, 0.6)--(9,0.6);
            \draw[EdgeHyper](6, -0.6)--(9,-0.6);
            \draw[EdgeHyper](6,0.6)arc(90:270:0.6);
            \draw[EdgeHyper](9,0.6)arc(90:-90:0.6);

            \draw[EdgeHyper](0, 0.7)--(4,0.7);
            \draw[EdgeHyper](0, -0.7)--(4,-0.7);
            \draw[EdgeHyper](0,0.7)arc(90:270:0.7);
            \draw[EdgeHyper](4,0.7)arc(90:-90:0.7);

            \draw[EdgeHyper](0, 0.8)--(9,0.8);
            \draw[EdgeHyper](0, -0.8)--(9,-0.8);
            \draw[EdgeHyper](0,0.8)arc(90:270:0.8);
            \draw[EdgeHyper](9,0.8)arc(90:-90:0.8);
        \end{tikzpicture}.
    \end{equation*}
    We also represent a $B$-tree obtained by beginning with the above choice:
    
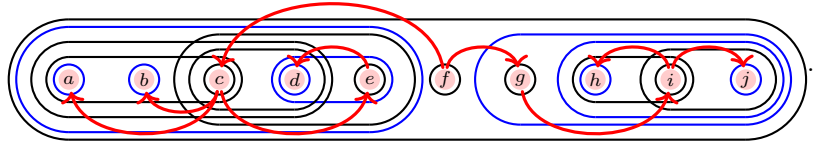
\begin{equation*}
        \begin{tikzpicture}[Centering]
            \node[NodeHyper, inner sep =0pt, draw=white, fill=Red!20](f)at(5,0){$f$};

            \node[NodeHyper, inner sep =0pt, draw=white, fill=Red!20](c)at(2,0){$c$};
            \node[NodeHyper, inner sep =0pt, draw=white, fill=Red!20](g)at(6,0){$g$};

            \node[NodeHyper, inner sep =0pt, draw=white, fill=Red!20](a)at(0,0){$a$};
            \node[NodeHyper, inner sep =0pt, draw=white, fill=Red!20](b)at(1,0){$b$};
            \node[NodeHyper, inner sep =0pt, draw=white, fill=Red!20](e)at(4,0){$e$};
            \node[NodeHyper, inner sep =0pt, draw=white, fill=Red!20](i)at(8,0){$i$};

            \node[NodeHyper, inner sep =0pt, draw=white, fill=Red!20](d)at(3,0){$d$};
            \node[NodeHyper, inner sep =0pt, draw=white, fill=Red!20](h)at(7,0){$h$};
            \node[NodeHyper, inner sep =0pt, draw=white, fill=Red!20](j)at(9,0){$j$};

            \draw[EdgeHyper, Blue](a)circle(0.2);
            \draw[EdgeHyper, Blue](b)circle(0.2);

            \draw[EdgeHyper](c)circle(0.2);
            \draw[EdgeHyper](e)circle(0.2);
            \draw[EdgeHyper](f)circle(0.2);
            \draw[EdgeHyper](g)circle(0.2);
            \draw[EdgeHyper](i)circle(0.2);

            \draw[EdgeHyper, Blue](j)circle(0.2);
            \draw[EdgeHyper, Blue](d)circle(0.2);
            \draw[EdgeHyper, Blue](h)circle(0.2);

            \draw[EdgeHyper](0, 0.3)--(2,0.3);
            \draw[EdgeHyper](0, -0.3)--(2,-0.3);
            \draw[EdgeHyper](0,0.3)arc(90:270:0.3);
            \draw[EdgeHyper](2,0.3)arc(90:-90:0.3);

            \draw[EdgeHyper, Blue](3, 0.3)--(4,0.3);
            \draw[EdgeHyper, Blue](3, -0.3)--(4,-0.3);
            \draw[EdgeHyper, Blue](3,0.3)arc(90:270:0.3);
            \draw[EdgeHyper, Blue](4,0.3)arc(90:-90:0.3);

            \draw[EdgeHyper](7, 0.3)--(8,0.3);
            \draw[EdgeHyper](7, -0.3)--(8,-0.3);
            \draw[EdgeHyper](7,0.3)arc(90:270:0.3);
            \draw[EdgeHyper](8,0.3)arc(90:-90:0.3);

            \draw[EdgeHyper](2, 0.4)--(3,0.4);
            \draw[EdgeHyper](2, -0.4)--(3,-0.4);
            \draw[EdgeHyper](2,0.4)arc(90:270:0.4);
            \draw[EdgeHyper](3,0.4)arc(90:-90:0.4);

            \draw[EdgeHyper](8, 0.4)--(9,0.4);
            \draw[EdgeHyper](8, -0.4)--(9,-0.4);
            \draw[EdgeHyper](8,0.4)arc(90:270:0.4);
            \draw[EdgeHyper](9,0.4)arc(90:-90:0.4);

            \draw[EdgeHyper](0, 0.5)--(3,0.5);
            \draw[EdgeHyper](0, -0.5)--(3,-0.5);
            \draw[EdgeHyper](0,0.5)arc(90:270:0.5);
            \draw[EdgeHyper](3,0.5)arc(90:-90:0.5);

            \draw[EdgeHyper, Blue](7, 0.5)--(9,0.5);
            \draw[EdgeHyper, Blue](7, -0.5)--(9,-0.5);
            \draw[EdgeHyper, Blue](7,0.5)arc(90:270:0.5);
            \draw[EdgeHyper, Blue](9,0.5)arc(90:-90:0.5);

            \draw[EdgeHyper](2, 0.6)--(4,0.6);
            \draw[EdgeHyper](2, -0.6)--(4,-0.6);
            \draw[EdgeHyper](2,0.6)arc(90:270:0.6);
            \draw[EdgeHyper](4,0.6)arc(90:-90:0.6);

            \draw[EdgeHyper, Blue](6, 0.6)--(9,0.6);
            \draw[EdgeHyper, Blue](6, -0.6)--(9,-0.6);
            \draw[EdgeHyper, Blue](6,0.6)arc(90:270:0.6);
            \draw[EdgeHyper, Blue](9,0.6)arc(90:-90:0.6);

            \draw[EdgeHyper, Blue](0, 0.7)--(4,0.7);
            \draw[EdgeHyper, Blue](0, -0.7)--(4,-0.7);
            \draw[EdgeHyper, Blue](0,0.7)arc(90:270:0.7);
            \draw[EdgeHyper, Blue](4,0.7)arc(90:-90:0.7);

            \draw[EdgeHyper](0, 0.8)--(9,0.8);
            \draw[EdgeHyper](0, -0.8)--(9,-0.8);
            \draw[EdgeHyper](0,0.8)arc(90:270:0.8);
            \draw[EdgeHyper](9,0.8)arc(90:-90:0.8);

            \draw[EdgeGraph, Red, ->](f)edge[bend right=80](c);
            \draw[EdgeGraph, Red, ->](f)edge[bend left=80](g);

            \draw[EdgeGraph, Red, ->](c)edge[bend left=80](a);
            \draw[EdgeGraph, Red, ->](c)edge[bend left=80](b);
            \draw[EdgeGraph, Red, ->](c)edge[bend right=80](e);
            \draw[EdgeGraph, Red, ->](g)edge[bend right=80](i);

            \draw[EdgeGraph, Red, ->](e)edge[bend right=80](d);
            \draw[EdgeGraph, Red, ->](i)edge[bend right=80](h);
            \draw[EdgeGraph, Red, ->](i)edge[bend left=80](j);
        \end{tikzpicture}.
    \end{equation*}
\end{example}

Let $V$ be a set, $F$ a rooted forest on $V$ and $c$ a coloring of $V$. We say that $F$ and $c$ are (strictly) compatible if $c$ is a
(strictly) increasing map on $F$. Respectively denote by $\scolorings{F,n}$ and $\colorings{F,n}$ the set of colorings strictly compatible
and compatible with $F$.

\begin{lemma}\label{lemma_bs}
    Let $V$ be a finite set and $B$ a building set on $V$. Then there is a bijection $\acyclic_B\cong \bfor{B}$ which preserves (strict)
    compatibility with colorings.
\end{lemma}

\begin{proof}
    First remark that for $B$ and $B'$ two building sets over two disjoint sets we have $\acyclic_{B\sqcup B'}=\acyclic_B\times\acyclic_{B'}$ 
    and $\bfor{B\sqcup B'}=\bfor{B}\times\bfor{B'}$, so it is sufficient to prove this lemma on connected building sets and $B$-trees.

    We construct a bijection $b$ between $\acyclic_B$ and $\bfor{B}$ by induction on $\card{V}$. If $V$ is a singleton then there is a unique
    $B$-tree possible and a unique acyclic orientation possible and the bijection is trivial. 
    Suppose now $V$ is not a singleton. Let $f$ be an acyclic orientation of $B$. Let $V_1,\dots V_k$ be the maximal connected sets not intersecting $f(V)$
    and $B_1,\dots,B_k$ their associated connected building sets. Then for $i\in [k]$, $f_{|V_i}$ is an acyclic orientation of $B_i$ and $T_i = b(f_{|V_i})$ 
    is a $B_i$-tree for $1\leq i\leq k$. Define $b(f)$ as the rooted tree in $f(V)$ obtained by doing the disjoint union of the $T_i$ and adding an edge between 
    $f(V)$ and the roots of the $T_i$s. Then $b(f)$ is a $B$-tree by Lemma~\ref{lemma_bfor}.

    Let now $T$ be a $B$-tree and $r$, $V_1,\dots, V_k$ and $B_1,\dots, B_k$ as defined in Lemma~\ref{lemma_bfor}. Let $B_0$ be the set of edges intersecting with
    $r$ and let $b^{-1}(T)$ be the orientation defined by $b^{-1}(T)(e) = b^{-1}(T_i)(e)$ if $e\in B_i$ and $b^{-1}(T)(e)=r$ if $e\in B_0$. Suppose there exists 
    $e_1,\dots,e_k$ a cycle in $b^{-1}(T)$. Since $r,V_1,\dots, V_k$ is partition of $V$, this cycle must be entirely contained in one of the $B_i$. Since by 
    induction this cycle can not be entirely in a $B_i$, it must be contained in $B_0$. But for every edge in $e\in B_0$ $b^{-1}(f)(e)=r$. Hence $b^{-1}(f)$ is acyclic.

    The fact that in the two preceding constructions the root of the $B$-tree is the image of the connected component along with the induction hypothesis 
    enable us to conclude that the $b$ and $b^{-1}$ thus defined are indeed inverse functions. It also gives us that $b$ preserves (strict) compatibility with 
    colorings, since for a (strictly) compatible coloring $c$ of $f$ the color $c(f(V))$ is necessarily the maximal color: 
    $f(V)=\max_c(V)=\set{v\in V\,|\, \text{$c(v)$ is maximal in $V$}}$.
\end{proof}

Let $B$ be a building set and $F$ be a $B$-forest. The {\em $F$-induced building set} is the building set whose connected sets are obtained by taking the 
non empty intersection of each connected set with the maximum node possible in $F$. We denote it by $B\cap F$.
For $\zeta$ a character of $\BS$ we say that $F$ is a {\em $\zeta$-$B$-forest} if $\zeta(B\cap F)\not=0$.

\begin{proposition}
    Let $\zeta$ be a character of $\BS$, $V$ be a finite set and $B\in\BS[V]$ be a building set. We then have:
    \begin{equation*}\begin{split}
        \chi^{\zeta}_V(B)(n) &= \sum_{F\in\bfor{B}} \zeta(B\cap F)\card{\scolorings{F,n}},\\
        \chi^{\zeta}_V(B)(-n) &= \sum_{F\in\bfor{B}} (-1)^{\cc(B\cap F)}\zeta(B\cap F)\card{\colorings{F,n}}.
    \end{split}\end{equation*}
    If $\zeta$ is a characteristic function, then $\chi^{\zeta}(B)(n)$ is the number of strictly compatible pairs of $\zeta$-$B$-forests
    and colorings with $[n]$. Furthermore, if $\zeta$ is odd $(-1)^{\card{V}}\chi^{\zeta}_V(B)(-n)$ is the number of compatible ones. In particular,
    $(-1)^{\card{V}}\chi^{\zeta}_V(B)(-1)$ is the number of $\zeta$-$B$-forest.
\end{proposition}

\begin{proof}
    Again, this is a direct application of the previous lemma. We just need to remark from the construction of the bijection of Lemma~\ref{lemma_bs} 
    that with our definition of $B\cap F$ we have $B\cap F=b^{-1}(F)(B)$.
\end{proof}


\subsection{Graphs, ripping and sewing}
The species $\SG$ of simple graphs admits another Hopf monoid structure than the one given in  
subsection~\ref{cor_graphs}. This Hopf monoid is defined in Definition 23.2 of \cite{AA} and its structure maps are given by
\begin{align*}
    \mu_{V_1,V_2}: \SG[V_1]\otimes \SG[V_2] &\to \SG[V] & \Delta_{V_1,V_2}: \SG[V] &\to \SG[V_1]\otimes \SG[V_2] \\
    g_1\otimes g_2 &\mapsto g_1\sqcup g_2 & g &\mapsto g|_{V_1}\otimes g/_{V_1}, \nonumber
\end{align*}
where $g|_{V_1}$ is the sub-graph of $g$ induced by $V_1$ and $g/_{V_1}$ is the simple graph on $V_2$ with an edge between $v$ and $v'$ if there 
is a path from $v$ to $v'$ in which all the vertices which are not ends are in $V_1$. These two operations are respectively called {\em ripping out $V_1$} 
and {\em sewing through $V_1$}.

Here is an example of co-product with the set $V=\set{a,b,c,d,e,f}$ and $V_1=\set{b,c,e}$ and $V_2=\set{a,d,f}$:
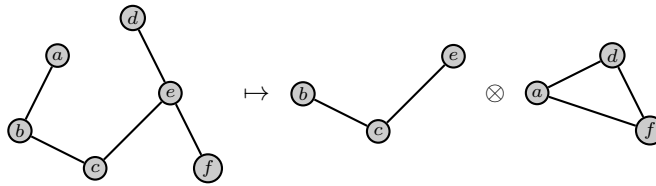
\begin{equation*}
    \begin{tikzpicture}[Centering,scale=1]
        \node[NodeHyper](c)at(0,0){$c$};
        \node[NodeHyper](a)at(-0.5,1.5){$a$};
        \node[NodeHyper](b)at(-1,0.5){$b$};
        \node[NodeHyper](d)at(0.5,2){$d$};
        \node[NodeHyper](e)at(1,1){$e$};
        \node[NodeHyper](f)at(1.5,0){$f$};
        \draw[EdgeHyper](a)--(b);
        \draw[EdgeHyper](b)--(c);
        \draw[EdgeHyper](c)--(e);
        \draw[EdgeHyper](e)--(d);
        \draw[EdgeHyper](e)--(f);
    \end{tikzpicture}
    \enspace \mapsto \enspace
    \begin{tikzpicture}[Centering,scale=1]
        \node[NodeHyper](c)at(0,0){$c$};
        \node[NodeHyper](b)at(-1,0.5){$b$};
        \node[NodeHyper](e)at(1,1){$e$};
        \draw[EdgeHyper](b)--(c);
        \draw[EdgeHyper](c)--(e);
    \end{tikzpicture}
    \enspace \otimes \enspace
        \begin{tikzpicture}[Centering,scale=1]
        \node[NodeHyper](a)at(0,0){$a$};
        \node[NodeHyper](d)at(1,0.5){$d$};
        \node[NodeHyper](f)at(1.5,-0.5){$f$};
        \draw[EdgeHyper](a)--(d);
        \draw[EdgeHyper](d)--(f);
        \draw[EdgeHyper](f)--(a);
    \end{tikzpicture}
\end{equation*}

\begin{definition}[Definition 23.1 in \cite{AA}]
    Let be $g\in \SG[V]$. A {\em tube} is a subset $W\subseteq V$ such that $g|_W$ is connected. The set of tubes of $g$ is a building set called 
    {\em graphical building set} of $g$ and which we denote $\text{tubes}(g)$.
\end{definition}

By Proposition 23.3 of \cite{AA} we know that $g\mapsto \text{tubes}(g)$ is a Hopf monoid morphism from $\SG$ to $\BS$.

\begin{definition}\label{def_part_tree}
    Let be $V$ be a finite set and $g\in \SG_c[V]$ a connected simple graph. A {\em partitioning tree} of $g$ is a rooted tree whose vertices are 
    labelled by the elements of a partition of $V$ and which is inductively defined by the following.
    \begin{itemize}
        \item If $V$ is a singleton, then the unique partitioning tree is the trivial tree with $V$ as sole vertex.
        \item Else choose $W\subset V$ and a partitioning tree for each connected component of $g_{V\setminus W}$. The tree with root $W$ and direct sub-trees these 
        partitioning trees is then a partitioning tree of $g$.
    \end{itemize}
    If $g$ is not connected anymore, a {\em partitioning forest of $g$} is the disjoint union of partitioning trees of each connected component of $g$.
    We denote by $\text{PF}(g)$ the set of partitioning forest of $g$.
\end{definition}

Let $g$ be a simple graph and $F$ be a partitioning forest of $g$. The graph {\em ripped and sewed through $F$} is the graph $g_T$ obtained
by the following process. Begin with $g_T=\emptyset$ and iteratively repeat the following: for each leaf $V$ of $F$, add $g_{V}$ to $g_F$ and sew $g$
through $V$. Delete all the leaves of $F$ and repeat the previous operation. The process terminates when $F$ is empty.
For $\zeta$ a character of $\SG$, we say that $F$ is a $\zeta$-partitioning forest if $\zeta(g_F)\not=0$.

\begin{proposition}
    Let $\zeta$ be a character of $\SG$, $V$ be a finite set and $g\in\SG[V]$ be a simple graph. We then have:
    \begin{equation*}\begin{split}
        \chi^{\zeta}_V(g)(n) &= \sum_{F\in\text{PF}(g)} \zeta(g_F)\card{\scolorings{F,n}},\\
        \chi^{\zeta}_V(g)(-n) &= \sum_{F\in\text{PF}(g)} (-1)^{\cc(g_F)}\zeta(g_F)\card{\colorings{F,n}}.
    \end{split}\end{equation*}
    If $\zeta$ is a characteristic function, then $\chi^{\zeta}(g)(n)$ is the number of strictly compatible pairs of $\zeta$-partitioning forests
    of $g$ and colorings with $[n]$. Furthermore, if $\zeta$ is odd $(-1)^{\card{V}}\chi^{\zeta}_V(g))(-n)$ is the number of compatible ones. In particular,
    $(-1)^{\card{V}}\chi^{\zeta}_V(g)(-1)$ is the number of $\zeta$-partitioning forests of $g$.
\end{proposition}

\begin{proof}
    Since $\chi^{\SG,\zeta}(g)(n) = \chi^{\BS,\zeta\circ\text{tubes}^{-1}}\circ\text{tubes}(g)(n)$, we only must verify that the partitioning forest
    of $g$ are the $\text{tubes}(g)$-forest and that $\text{tubes}(g_F) = \text{tubes}(g)\cap F$. We do this in the case where $g$ is connected, the other
    case being a direct consequence from this one. This is quite straightforward: in the definition of a partitioning tree, the set $W$ is a subset of 
    $V=\text{tubes}(g)$ and if $V_1,\dots, V_k$ are the set of vertices of the connected components of $g_{V\setminus W}$, then there are the maximal
    connected sets of $\text{tubes}(g)$ which does not intersect with $g$. This and Lemma~\ref{lemma_bfor} show that the definition of partitioning trees
    of $g$ is the same than that of $\text{tubes}(g)$-trees and hence they are the same objects.

    Let us now show that $\text{tubes}(g_T)=\text{tubes}(g)\cap T$. Let $V_1,\dots, V_k$ be the nodes of $T$ starting with the leaves and going the way up
    until the root. Denote by $D$ the partition $V_1,\dots, V_k$. Then by definition of $g_T$, $g_T=\mu_D\circ\Delta_D(g)$. Since $\text{tubes}$ is a Hopf 
    monoid morphism, we have that $\text{tubes}(mu_D\circ\Delta_D(g))=\mu_D\circ\Delta_D(\text{tubes}(g))=\max_D(\text{tubes}(g))=b(T)(\text{tubes}(g))
    =\text{tubes}(g)\cap T$, where $\max_D$ is the unique strictly compatible orientation of $D$ and $b$ is the bijection defined in Lemma~\ref{lemma_bs}.
\end{proof}

We have a particular interpretation of this proposition for the basic character.

\begin{corollary}[Corollary 34 in \cite{my1}]
    Let $V$ be a finite set and $g\in\SG[V]$ a simple graph. Then $\chi^{\zeta_1}_V(g)(n)$ is the number of colorings of $V$ with $[n]$ such that every path 
    in $g$ with ends of the same color has a vertex of color strictly greater than the colors of the ends.
\end{corollary}

\begin{proof}
    Again we only show this in the case of $g$ connected. Since $\zeta_1$ is the characteristic function of discrete elements, the $\zeta_1$-partitioning trees are 
    the trees where every node is a singleton. Let $T$ be a partitioning tree of $g$ and $c$ a coloring strictly compatible with $T$. Let $\set{v}$ and $\set{v'}$ 
    be two nodes of $T$ of the same color. Then they are incomparable by definition of strict compatibility. Let $\set{a}$ be their lowest common ancestor and $A$
    the set of strict ancestors of $a$ ($a\not\in A$). Again by definition of strict compatibility, $a$ has a color strictly greater than $v$ and $v'$. Every 
    path from $v$ to $v'$ must pass by $a$ since $v$ and $v'$ are in two different connected components of $g_{V\setminus A\cup\set{a}}$ be in a same connected 
    component of $g_{V\setminus A}$. 

    Let now $c$ be a coloring of $V$ such that every path in $g$ with ends of the same color has a vertex of color strictly greater than the colors of the ends.
    Then there is a unique vertex in $c$ which is of maximal color: else, since $g$ is connected, we would have a path between two of them and hence a vertex with a
    color strictly greater. Let $v$ be the unique vertex of maximal color. Then each connected component of $g_{V\setminus v}$ must also have one vertex of maximal
    color. Hence we have an inductive way to construct a tree strictly compatible with $c$. Since this way coincides with the definition of partitioning tree whose
    vertices are singletons, this concludes the proof.
\end{proof}

\subsection{Partitions}
Recall that a {\em partition of $V$} is a set $\set{P_1,\dots, P_k}$ of disjoint non empty sets, called {\em parts}, such that $\sqcup_i P_i = V$. The linear 
species $\Part$ admits a Hopf monoid structure with structure maps:
\begin{align*}
    \mu_{V_1,V_2}: \Part[V_1]\otimes \Part[V_2] &\to \Part[V] & \Delta_{V_1,V_2}: \Part[V] &\to \Part[V_1]\otimes \Part[V_2] \\
    \pi_1\otimes \pi_2 &\mapsto \pi\sqcup \pi_2 & \pi &\mapsto \pi|_{V_1}\otimes \pi|_{V_2}, \nonumber
\end{align*}
where for $\pi=\set{P_1,\dots,P_k}$, $\pi|_{V_1}$ is the partition of $V_1$ obtained by taking the intersection with $V_1$ of each part $P_i$ and forgetting 
the empty parts.

A {\em cliquey graph} is a simple graph which is the disjoint union of cliques. The species morphism which sends a partition $\set{P_1,\dots,P_k}$ on the 
cliquey graph composed of the cliques on $P_1,\dots,P_k$. By Proposition 24.2 of \cite{AA}, this is a Hopf monoid morphism $\Part\to\SG^{cop}$ with the 
ripping and sewing Hopf structure on $\SG$.

We say that a partition $\tau$ refines a partition $\pi$, and denote $\tau\prec\pi$ is the parts of $\pi$ are the unions of the parts of $\tau$.
For $D$ a decomposition, denote by $\pi(D)$ the partition obtained by forgetting the order and the empty parts of $D$. We say that a decomposition $D$ is 
induced by $\pi$ is $\pi(D)\prec\pi$ 
For $\zeta$ a character of $\Part$, we say that a decomposition $D$ is a $\zeta$-decomposition if $\zeta(\pi(D))\not=0$.

\begin{proposition}
    Let $\zeta$ be a character of $\Part$, $V$ be a finite set and $\pi\in\Part[V]$ be a partition. We then have:
    \begin{equation*}\begin{split}
        \chi^{\zeta}_V(\pi)(n) &= \sum_{\tau\prec\pi}\zeta(\tau) \ell(\tau)!\binom{n}{\ell(\tau)}
    \end{split}\end{equation*}
    If $\zeta$ is a characteristic function, then $\chi^{\zeta}(\pi)(n)$ is the number of $\zeta$-decomposition of size $n$ induced by $\pi$. 
\end{proposition}

\begin{proof}
    As done in the previous section, since $\chi$ is multiplicative we only need to look at the trivial partition $\pi=\set{V}$. Let $g$ be the image of $\pi$
    by the previously defined morphism \textit{i.e.}  $g$ is the cliquey graph over $V$. For $W\subset V$, it is clear that $g_{V\setminus W}$ is the cliquey graph over
    $V\setminus W$ and the partitioning trees of $g$ are then chain trees. These are in bijection with compositions of $V$. By definition of $g_T$, for $T=T_1,\dots,T_k$ 
    such a tree/decomposition, $g_T$ is the disjoint union of the cliquey graphs over $V_1,\dots,V_k$. A coloring strictly compatible with a line tree with $k$ 
    vertices is an increasing map from $[k]$ to $[n]$. This is equal to $\binom{n}{k}$ and is the number of decomposition of size $n$ which reduce to the line
    tree when forgetting the empty parts. Hence we have $\chi^{\zeta}_V(\set{V})(n) = \sum_{P\vDash V}\zeta(\pi(P))\binom{n}{\ell(P)}$ which is equal to the
    desired formula grouping by partitions $\tau$ such that $\pi(P)=\tau$. 
\end{proof}

We do not give the formula for the negative integers here since the formula for the non negative ones is sufficiently explicit and the objects are simple
enough that the notion of compatible colorings is not particularly revealing. 


\subsection{Set of paths}
A {\em word} on $V$ is a total ordering of $V$. The {\em paths} on $V$ are the words on $V$ quotiented by the relation $w_1\dots w_{|V|} \sim w_{|V|}\dots w_1$. 
A {\em set of paths} $\alpha$ of $V$ is a partition $\set{V_1,\dots, V_k}$ of $V$ with a path $s_i$ on each part $V_i$ and we will write $\alpha = s_1|\dots|s_l$. 
We denote by $\Path$ the species of set of paths. The species $\Path$ of sets of paths admits a Hopf monoid structure with structure maps
\begin{align*}
    \mu_{V_1,V_2}: \Path[V_1]\otimes \Path[V_2] &\to \Path[V] & \Delta_{V_1,V_2}: \Path[V] &\to \Path[V_1]\otimes \Path[V_2] \\
    \alpha_1\otimes \alpha_2 &\mapsto \alpha_1\sqcup \alpha_2 & \alpha &\mapsto \alpha|_{V_1}\otimes \alpha/_{V_1} \nonumber
\end{align*}
where if $\alpha = s_1|\dots|s_l$, $\alpha|_{V_1}= s_1\cap S|\dots| s_l\cap S$ forgetting the empty parts and $\alpha/_{V_1}$ is the set of paths obtained by 
replacing each occurrence of an element of $V_1$ in $\alpha$ by the separation symbol $|$ and removing the multiplicity of $|$. 

\begin{example}
    For $V=\set{a,b,c,d,e,f,g}$ and $V_1=\set{b,c,e}$ and $V_2=\set{a,d,f,g}$, we have:
    \begin{equation*}
        \Delta_{V_1,V_2}(bfcg|aed) = bc|e\otimes f|g|a|d.
    \end{equation*}
\end{example}

For $\alpha=s_1|\dots|s_l$ a set of path, denote by $l(\alpha)$ the simple graph whose connected components are the paths induced by $s_1,\dots, s_l$. 
By Proposition 25.1 of \cite{AA} we know that $\alpha \mapsto l(\alpha)$ is a morphism of Hopf monoids from $\Path$ to $\SG^{cop}$.

\begin{example}
    For $V=\{a,b,c,d,e,f,g\}$ and $\alpha=bfcg|aed$, $l(\alpha)$ is the following graph:
    
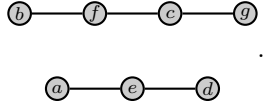
\begin{equation*}
        \begin{tikzpicture}[Centering, scale=1]
            \node[NodeHyper, inner sep = 0pt](a)at(0,0){$a$};
            \node[NodeHyper, inner sep = 0pt](b)at(-0.5,1){$b$};
            \node[NodeHyper, inner sep = 0pt](c)at(1.5,1){$c$};
            \node[NodeHyper, inner sep = 0pt](d)at(2,0){$d$};
            \node[NodeHyper, inner sep = 0pt](e)at(1,0){$e$};
            \node[NodeHyper, inner sep = 0pt](f)at(0.5,1){$f$};
            \node[NodeHyper, inner sep = 0pt](g)at(2.5,1){$g$};
            \draw[EdgeHyper](a)--(e);
            \draw[EdgeHyper](e)--(d);
            \draw[EdgeHyper](b)--(f);
            \draw[EdgeHyper](f)--(c);
            \draw[EdgeHyper](c)--(g);
        \end{tikzpicture}.
    \end{equation*}
\end{example}

We only give the interpretation for $\chi^{\zeta_1}$ here.

\begin{proposition}[Corollary 38 in \cite{my1}]
    Let $V$ be a finite set and $\alpha\in\Path[V]$ be a path on $V$. Then $\chi^{\zeta_1}_V(\alpha)(n)$ is the number of strictly compatible pairs of binary trees 
    with $\card{V}$ vertices and colorings with $[n]$ and $\chi^{\zeta_1}_V(\alpha)(-n)$ is the number of compatible pairs of binary trees with $\card{V}$ vertices 
    and colorings with $[n]$. In particular $\chi^{\zeta_1}_V(\alpha)(-1) = C_{\card{V}}$ where $C_n = \frac{1}{n+1}\binom{2n}{n}$ is the $n$-th Catalan number.
\end{proposition}

\begin{proof}
    First remark that by definition of $\chi$, we have $\chi^{\SG^{cop}} = \chi^{\SG}$ and so \linebreak[4] $\chi^{\Path,\zeta_1}_V(\alpha)(n) = \chi^{\SG,\zeta_1}_V(l(\alpha))(n)$. 
    Fix one of the two total orderings of $V$ induced by $\alpha$ so that we can consider the left and the right of a vertex $v$ of $l(\alpha)$. Then each 
    vertex of $l(\alpha)$ is totally characterised by the number of vertices on its left (and on its right) and hence the partitioning trees of $l(\alpha)$ 
    with singletons for vertices are exactly the binary trees with $\card{V}$ vertices.
\end{proof}


\bibliographystyle{plain}
\bibliography{redac}

\end{document}